\documentclass[11pt]{llncs}

\usepackage[margin=1in]{geometry}
\usepackage{authblk}
\usepackage{amsmath, amssymb}
\usepackage{xspace}
\usepackage{dsfont}
\usepackage[colorlinks=true,urlcolor=webblue,linkcolor=webgreen,filecolor=webblue,citecolor=webgreen,pdfpagemode=UseOutlines,pdfstartview=FitH,pdfpagelayout=OneColumn,bookmarks=true]{hyperref}
\usepackage{enumerate}

\usepackage[ruled,vlined]{algorithm2e}
\SetKwInput{KwInput}{Input}
\SetKwInput{KwOutput}{Output}

\usepackage[backend=biber,style=alphabetic,sorting=anyt,doi=false,isbn=false,url=false,maxbibnames=20,
maxalphanames=3,minalphanames=3 % ensures 3 capital leters in citation labels
]{biblatex}
\bibliography{references.bib}
\newbibmacro{doiurl}[1]{%
  \iffieldundef{doi}{%
    \iffieldundef{url}{#1}{%
      \href{\thefield{url}}{#1}%
    }%
  }{\href{https://dx.doi.org/\thefield{doi}}{#1}}%
}
\newcommand{\papertitle}[1]{\mkbibquote{#1}}
\DeclareFieldFormat{title}{\usebibmacro{doiurl}{\papertitle{#1}}}
\DeclareFieldFormat[article]{title}{\usebibmacro{doiurl}{\papertitle{#1}}}
\DeclareFieldFormat[incollection]{title}{\usebibmacro{doiurl}{\papertitle{#1}}}
\DeclareFieldFormat[inproceedings]{title}{\usebibmacro{doiurl}{\papertitle{#1}}}
\DeclareFieldFormat[inbook]{title}{\usebibmacro{doiurl}{\papertitle{#1}}}
\DeclareFieldFormat[book]{title}{\usebibmacro{doiurl}{\papertitle{#1}}}
\DeclareFieldFormat[thesis]{title}{\usebibmacro{doiurl}{\papertitle{#1}}}
\DeclareFieldFormat{eprint:iacr}{Cryptology ePrint Archive: \href{https://eprint.iacr.org/#1}{#1}}

\usepackage{tikz}

\definecolor{webgreen}{rgb}{0,.5,0}
\definecolor{webblue}{rgb}{0,0,.5}

\newtheorem{construction}{Construction}
\usepackage{mathtools}
\DeclarePairedDelimiter{\set}{\lbrace}{\rbrace}
\DeclarePairedDelimiter{\abs}{\lvert}{\rvert}
\DeclarePairedDelimiter{\norm}{\lVert}{\rVert}

\DeclarePairedDelimiter{\ip}{\langle}{\rangle}
\DeclarePairedDelimiter{\floor}{\lfloor}{\rfloor}
\DeclarePairedDelimiter{\ceil}{\lceil}{\rceil}
\DeclarePairedDelimiter{\nint}{\lfloor}{\rceil}

\newcommand{\Z}{\mathbb{Z}}

\renewcommand{\vec}[1]{\boldsymbol{#1}}  % Bold vectors instead of arrow vectors

\newcommand{\ket}[1]{| #1 \rangle}

\newcommand{\braket}[2]{\langle #1 | #2 \rangle}

\newcommand{\tr}{\mathrm{Tr}}

\newcommand{\tp}{^{\mathsf{T}}}

\renewcommand{\rho}{\varrho}

\newcommand{\QFT}{\ensuremath{\mathsf{QFT}}\xspace}

\newcommand{\PKE}{\ensuremath{\mathsf{PKE}}\xspace}

\newcommand{\PRFscheme}{\ensuremath{\mathsf{PRFscheme}}\xspace}
\newcommand{\PRPscheme}{\ensuremath{\mathsf{PRPscheme}}\xspace}
\newcommand{\LWESKES}{\ensuremath{\mathsf{LWE\mbox{-}SKE}}\xspace}
\newcommand{\LWEPKES}{\ensuremath{\mathsf{LWE\mbox{-}PKE}}\xspace}

\newcommand{\FrodoKEM}{\ensuremath{\mathsf{FrodoKEM}}\xspace}
\newcommand{\FrodoPKE}{\ensuremath{\mathsf{FrodoPKE}}\xspace}

{\begin{framed}\begin{small}}
{\end{small}\end{framed}}

\newcommand{\expref}[2]{\texorpdfstring{\hyperref[#2]{#1~\ref{#2}}}{#1~\ref{#2}}}

\newcommand{\A}{\ensuremath{\mathcal{A}}\xspace}

\newcommand{\negl}{\ensuremath{\operatorname{negl}}\xspace}
\newcommand{\from}{\ensuremath{\leftarrow}}
\newcommand{\bit}{\{0,1\}}
\newcommand{\pk}{\ensuremath{pk}\xspace}
\newcommand{\sk}{\ensuremath{sk}\xspace}
\newcommand{\rand}{\raisebox{-1pt}{\ensuremath{\,\xleftarrow{\raisebox{-1pt}{$\scriptscriptstyle\$$}}\,}}}
\newcommand{\randchi}{\raisebox{-1pt}{\ensuremath{\,\xleftarrow{\raisebox{-1pt}{$\scriptscriptstyle\chi$}}\,}}}

\DeclareMathOperator{\E}{\mathbb{E}}

\newcommand{\KeyGen}{\ensuremath{\mathsf{KeyGen}}\xspace}
\newcommand{\Enc}{\ensuremath{\mathsf{Enc}}\xspace}
\newcommand{\Dec}{\ensuremath{\mathsf{Dec}}\xspace}

\newcommand{\poly}{\operatorname{poly}}
\newcommand{\algo}{\mathcal}
\newcommand{\inrand}{\raisebox{-1pt}{\ensuremath{\,\xleftarrow{\raisebox{-1pt}{$\scriptscriptstyle\$$}}\,}}}

\newcommand{\PRF}{\ensuremath{\mathsf{PRF}}\xspace}

\newcommand{\PRP}{\ensuremath{\mathsf{PRP}}\xspace}
\newcommand{\QPRF}{\ensuremath{\mathsf{QPRF}}\xspace}

\newcommand{\QPRP}{\ensuremath{\mathsf{QPRP}}\xspace}

\newcommand{\QPT}{\ensuremath{\mathsf{QPT}}\xspace}

\newcommand{\Round}{\ensuremath{\mathsf{LRF}}\xspace}

\newcommand{\RA}{\ensuremath{\mathsf{RA}}\xspace} % randomness-access

\newcommand{\LWE}{\ensuremath{\mathsf{LWE}}\xspace}
\newcommand{\RingLWE}{\ensuremath{\mathsf{Ring\mbox{-}LWE}}\xspace}
\newcommand{\QRAC}{\ensuremath{\mathsf{QRAC}}\xspace}

\newcommand{\PPT}{\ensuremath{\mathsf{PPT}}\xspace}

\newcommand{\CPA}{\ensuremath{\mathsf{CPA}}\xspace}
\newcommand{\CCa}{\ensuremath{\mathsf{CCA}}\xspace}
\newcommand{\CCA}{\ensuremath{\mathsf{CCA1}}\xspace}
\newcommand{\CCAA}{\ensuremath{\mathsf{CCA2}}\xspace}

\newcommand{\INDCPA}{\ensuremath{\mathsf{IND\mbox{-}CPA}}\xspace}

\newcommand{\INDCCAA}{\ensuremath{\mathsf{IND\mbox{-}CCA2}}\xspace}

\newcommand{\INDQCPA}{\ensuremath{\mathsf{IND\mbox{-}QCPA}}\xspace}
\newcommand{\INDQCCA}{\ensuremath{\mathsf{IND\mbox{-}QCCA1}}\xspace}

\newcommand{\INDQCCAA}{\ensuremath{\mathsf{IND\mbox{-}QCCA2}}\xspace}

\newcommand{\QCPA}{\ensuremath{\mathsf{QCPA}}\xspace}
\newcommand{\QCCA}{\ensuremath{\mathsf{QCCA1}}\xspace}
\newcommand{\QCCAA}{\ensuremath{\mathsf{QCCA2}}\xspace}

\newcommand{\Samp}{\ensuremath{\mathsf{Samp}}\xspace}

\newcommand{\SEMQCPA}{\ensuremath{\mathsf{SEM\mbox{-}QCPA}}\xspace}
\newcommand{\SEMQCCA}{\ensuremath{\mathsf{SEM\mbox{-}QCCA1}}\xspace}
\newcommand{\SEMQCCAA}{\ensuremath{\mathsf{SEM\mbox{-}QCCA2}}\xspace}

\newcommand{\IndGame}{\ensuremath{\mathsf{IndGame}}\xspace}
\newcommand{\IndGameR}{\ensuremath{\mathsf{IndGame}'}\xspace}
\newcommand{\SemGame}{\ensuremath{\mathsf{SemGame}}\xspace}

\newcommand{\real}{\ensuremath{\mathsf{real}}\xspace}
\newcommand{\simul}{\ensuremath{\mathsf{sim}}\xspace}

\newcommand{\phase}[1]{\emph{#1:}}

\begin{document}

\vspace{-12pt}
\title{On Quantum Chosen-Ciphertext Attacks\\ and Learning with Errors}
\pagestyle{plain}
%\date{\today}
\author{Gorjan Alagic\inst{1} \and Stacey Jeffery\inst{2} \and Maris Ozols\inst{3} \and Alexander Poremba\inst{4}}
\institute{QuICS, University of Maryland, and NIST, Gaithersburg, MD, USA\and QuSoft and CWI, Amsterdam, Netherlands\and QuSoft and University of Amsterdam, Amsterdam, Netherlands\and Computing and Mathematical Sciences, Caltech, Pasadena, CA, USA}

\maketitle

\begin{abstract}
Large-scale quantum computing is a significant threat to classical public-key cryptography. In strong ``quantum access'' security models, numerous symmetric-key cryptosystems are also vulnerable. We consider classical encryption in a model which grants the adversary quantum oracle access to encryption and decryption, but where the latter is restricted to non-adaptive (i.e., pre-challenge) queries only. We define this model formally using appropriate notions of ciphertext indistinguishability and semantic security (which are equivalent by standard arguments) and call it \QCCA  in analogy to the classical \CCA security model. Using a bound on quantum random-access codes, we show that the standard \PRF- and \PRP-based encryption schemes are \QCCA-secure when instantiated with quantum-secure primitives. %a quantum-secure \PRF (pseudorandom function) or \PRP (pseudorandom permutation), respectively.

We then revisit standard \INDCPA-secure Learning with Errors (\LWE) encryption and show that leaking just one quantum decryption query (and no other queries or leakage of any kind) allows the adversary to recover the full secret key with constant success probability. In the classical setting, by contrast, recovering the key uses a linear number of decryption queries, and this is optimal. % can produce at most one bit of the $(n\log q)$-bit key.
%, although a polynomial number of classical decryption queries is enough to recover the key.
The algorithm at the core of our attack is a (large-modulus version of) the well-known Bernstein-Vazirani algorithm.
%Likewise, leaking just one quantum \emph{encryption} query, where the adversary is allowed to instantiate the randomness register, also leads to complete key recovery. By contrast, a classical decryption (respectively encryption) query can produce at most one bit (respectively $\log q$ bits) of the $(n\log q)$-bit key.
We emphasize that our results should \textbf{not} be interpreted as a weakness of these cryptosystems in their stated security setting (i.e., post-quantum chosen-plaintext secrecy). Rather, our results mean that, if these cryptosystems are exposed to chosen-ciphertext attacks (e.g., as a result of deployment in an inappropriate real-world setting) then quantum attacks are even more devastating than classical ones.
\end{abstract}

%%%%%%%%%%%%%%%%%%%%%%%%%
\section{Introduction}
%%%%%%%%%%%%%%%%%%%%%%%%%

\subsection{Background}
%%%%%%%%%%%%%%%%%%%%%%%%%

Large-scale quantum computers pose a dramatic threat to classical cryptography. The ability of such devices to run Shor's efficient quantum factoring algorithm (and its variants) would lead to devastation of the currently deployed public-key cryptography infrastructure~\cite{NIST16,Shor94}. This threat has led to significant work on so-called ``post-quantum'' alternatives, where a prominent category is occupied by cryptosystems based on the \emph{Learning with Errors} (\LWE) problem of solving noisy linear equations over $\Z_q$~\cite{Regev05} and its variants~\cite{NIST16,NISTPQC}. %The \LWE problem is widely believed to be intractable even for quantum computers, and thus forms the basis for a number of candidate post-quantum cryptosystems.

In addition to motivating significant work on post-quantum cryptosystems, the threat of quantum computers has also spurred general research on secure classical cryptography in the presence of quantum adversaries. One area in particular explores strong security models where a quantum adversary gains precise quantum control over portions of a classical cryptosystem. In such models, a number of basic symmetric-key primitives can be broken by simple quantum attacks based on Simon's algorithm~\cite{KM10,KM12,KLLN16,SS16,Simon97}. It is unclear if the assumption behind these models is plausible for typical physical implementations of symmetric-key cryptography. However, attacks which involve quantumly querying a classical function are always available in scenarios where the adversary has access to a circuit for the relevant function. This is the case for hashing, public-key encryption, and circuit obfuscation. Moreover, understanding this model is crucial for gauging the degree to which any physical device involved in cryptography must be resistant to reverse engineering or forced quantum behavior (consider, e.g., the so-called ``frozen smart card'' example~\cite{GHS16}). For instance, one may reasonably ask: \emph{what happens to the security of a classical cryptosystem when the device leaks only a single quantum query to the adversary?}

When deciding which functions the adversary might have (quantum) access to, it is worth recalling the classical setting. For classical symmetric-key encryption, a standard approach considers the security of cryptosystems when exposed to so-called chosen-plaintext attacks (\CPA). This notion encompasses all attacks in which an adversary attempts to defeat security (by, e.g., distinguishing ciphertexts or extracting key information) using oracle access to the function which encrypts plaintexts with the secret key. This approach has been highly successful in developing cryptosystems secure against a wide range of realistic real-world attacks. An analogous class, the so-called chosen-ciphertext attacks (\CCa), are attacks in which the adversary can make use of oracle access to decryption.
%A devastating attack in which an adversary exercises control over a cryptosystem is described by a chosen-ciphertext attack. In this situation, the attacker receives \emph{adaptive} access to the encryption and decryption procedures of a cryptosystem, i.e. the attacker is able to choose plaintexts and ciphertexts that depend on previous outcomes of the attack. In fact, such attacks are often perfectly realistic, and even widely adopted cryptosystems such as RSA have been shown vulnerable against variants of chosen-ciphertext attacks.
For example, a well-known attack due to Bleichenbacher \cite{Bleichenbacher98chosenciphertext}
only requires access to an oracle that decides if the input ciphertext is encrypted according to a particular RSA standard. We will consider analogues of both \CPA and \CCa attacks, in which the relevant functions are quantumly accessible to the adversary.

Prior works have formalized the quantum-accessible model for classical cryptography in several settings, including unforgeable message authentication codes and digital signatures~\cite{BZ13,BZ13a}, encryption secure against quantum chosen-plaintext attacks (\QCPA)~\cite{BJ15,GHS16}, and encryption secure against \emph{adaptive} quantum chosen-ciphertext attacks (\QCCAA)~\cite{BZ13}. %Gagliardoni, H\"ulsing and Schaffner also consider a wide range of indistinguishability and semantic-security models within the \QCPA framework~\cite{GHS16}.

\subsection{Our Contributions}
%%%%%%%%%%%%%%%%%%%%%%%%%

\paragraph{The model.}
%%%

In this work, we consider a quantum-secure model of encryption called \QCCA. This model grants \emph{non-adaptive} access to the decryption oracle, and is thus intermediate between \QCPA and \QCCAA. Studying weaker and intermediate models is a standard and useful practice in theoretical cryptography. In fact, \CPA and \textsf{CCA2} are intermediate models themselves, since they are both strictly weaker than authenticated encryption. Our particular intermediate model is naturally motivated: it is sufficent for a new and interesting quantum attack on \LWE encryption.

As is typical, the challenge in \QCCA can be semantic, or take the form of an indistinguishability test. %where the adversary supplies two challenge plaintexts $(m_0, m_1)$, receives a challenge ciphertext $\Enc_k(m_b)$ for random $b$, and must correctly guess $b$. Alternatively, the challenge can be semantic, where the adversary receives partial information about a plaintext $m$, and is tasked with outputting some \emph{additional} information about $m$ by making use of its encryption $\Enc_k(m)$. 
This leads to natural security notions for symmetric-key encryption, which we call \INDQCCA and \SEMQCCA, respectively. Following previous works, it is straightforward to define both \INDQCCA and \SEMQCCA formally, and prove that they are equivalent~\cite{BJ15,GHS16,BZ13}.

We then prove \INDQCCA security for two symmetric-key encryption schemes, based on standard assumptions. Specifically, we show that the standard encryption schemes based on quantum-secure pseudorandom functions (\QPRF) and quantum-secure pseudorandom permutations (\QPRP) are both \INDQCCA. We remark that both \QPRF{s} and \QPRP{s} can be constructed from quantum-secure one-way functions~\cite{Zha2012,Zha16}. Our security proofs use a novel technique, in which we control the amount of information that the adversary can extract from the oracles and store in their internal quantum state (prior to the challenge) by means of a certain bound on quantum random-access codes.

\paragraph{A quantum-query attack on \LWE.}
%%%

We then revisit the aforementioned question: what happens to a post-quantum cryptosystem if it leaks a single quantum query? Our main result is that standard \INDCPA-secure \LWE-based encryption schemes can be completely broken using only \emph{a single quantum decryption query} and no other queries or leakage of any kind. In our attack, the adversary recovers the complete secret key with constant success probability. In standard bit-by-bit \LWE encryption, a single classical decryption query can yield at most one bit of the secret key; the classical analogue of our attack thus requires $n \log q$ queries. The attack is essentially an application of a modulo-$q$ variant of the Bernstein-Vazirani algorithm~\cite{BV97}. Our new analysis shows that this algorithm correctly recovers the key with constant success probability, despite the decryption function only returning an inner product which is rounded to one of two values. We show that the attack applies to four variants of standard \INDCPA-secure \LWE-based encryption: the symmetric-key and public-key systems originally described by Regev~\cite{Regev05}, the \FrodoPKE scheme\footnote{\FrodoPKE is an \INDCPA-secure building block in the \INDCCAA-secure post-quantum cryptosystem ``FrodoKEM''~\cite{FrodoKEM}. Our results do not affect the post-quantum security of Frodo and do not contradict the \CCAA security of FrodoKEM.}~\cite{LP11,FrodoKEM}, and standard \RingLWE~\cite{LPR-toolkit-2013,LPR13}.

\paragraph{Important caveats.}
%%%

Our results challenge the idea that \LWE is unconditionally ``just as secure'' quantumly as it is classically. Nonetheless, the reader is cautioned to interpret our work carefully. Our results do \emph{not} indicate a weakness in \LWE (or any \LWE-based cryptosystem) in the standard post-quantum security model. Since it is widely believed that quantum-algorithmic attacks will need to be launched over purely classical channels, post-quantum security does not allow for quantum queries to encryption or decryption oracles. Moreover, while our attack does offer a dramatic quantum speedup (i.e., one query vs. linear queries), the classical attack is already efficient. The schemes we attack are already insecure in the classical chosen-ciphertext setting, but can be modified to achieve chosen-ciphertext security~\cite{FO99}. %\snote{edited and added citation}

%the plain \LWE-based encryption schemes we consider in this work are already vulnerable to decryption queries in a purely classical attack model. In fact, full key recovery against these schemes is already possible using an approximately linear number of classical decryption queries. Our results should thus not be interpreted as a weakness of these cryptosystems in their stated security setting (i.e., \INDCPA). The proper interpretation is that, if these cryptosystems are exposed to chosen-ciphertext attacks, then quantum attacks can be even more devastating than classical ones.

%In particular, we remark that one classical decryption query can produce at most one bit of the $(n \log q)$-bit key. Similarly, a classical encryption query which can access also the randomness used in the encryption can produce at most $\log q$ bits of the key.

%\ga{We should make sure the caveats are completely clear.}

%\ga{Another point of view on this result: we should stop viewing LWE as basically ``just as secure'' quantumly as it is classically -- even if that statement is conditional on our best quantum algorithms for solving ``the standard LWE problem'' being no better than the classical ones! Sure, for now, the setting in which we discovered a significant difference between LWE's quantum and classical security is not such a big deal... but it does indicate that there is a difference, and more settings like this could be found. }

%\ap{I added a few paragraphs above to address your comments Gorjan. Thoughts?}
\paragraph{Related work.}
%%%%%%%%%%%%%%%%%%%%%%%%%

We remark that Grilo, Kerenidis and Zijlstra recently observed that a version of \LWE with so-called ``quantum samples'' can be solved efficiently (as a learning problem) using Bernstein-Vazirani~\cite{GK17}. Our result, by contrast, demonstrates an actual cryptographic attack on standard cryptosystems based on \LWE, in a plausible security setting. Moreover, in terms of solving the learning problem, our analysis shows that constant success probability is achievable with only a single query, whereas~\cite{GK17} require a number of queries which is at least linear in the modulus $q$. In particular, our cryptographic attack succeeds with a single query even for superpolynomial modulus.

\subsection{Technical summary of results}
%%%%%%%%%%%%%%%%%%%%%%%%%

%We now outline our results with some further technical details.

\subsubsection{Security model and basic definitions.}
%%%%%%%%%%%%%%%%%%%%%%%%%

First, we set down the basic \QCCA security model, adapting the ideas of~\cite{BZ13a,GHS16}. Recall that an encryption scheme is a triple $\Pi = (\KeyGen, \Enc, \Dec)$ of algorithms (key generation, encryption, and decryption, respectively) satisfying $\Dec_k(\Enc_k(m)) = m$ for any key $k \from \KeyGen$ and message $m$.
%Here $\KeyGen$ represents a procedure for generating a key, $\Enc$ represents the encryption procedure, and $\Dec$ the decryption procedure.
In what follows, all oracles are quantum, meaning that a function $f$ is accessed via the unitary operator $\ket{x}\ket{y} \mapsto \ket{x}\ket{y \oplus f(x)}$. We define ciphertext indistinguishability and semantic security as follows.

\begin{definition}[informal] $\Pi$ is \INDQCCA if no quantum polynomial-time algorithm (\QPT) $\A$ can succeed at the following experiment with probability better than $1/2 + \negl(n)$.
\begin{enumerate}
\item A key $k \from \KeyGen(1^n)$ and a uniformly random bit $b \inrand \bit$ are generated;
$\algo A$ gets access to oracles $\Enc_k$ and $\Dec_k$, and outputs $(m_0, m_1)$;
\item $\algo A$ receives $\Enc_k(m_b)$ and gets access to an oracle for $\Enc_k$ only, and outputs a bit $b'$; $\algo A$ wins if $b = b'$.
\end{enumerate}
\end{definition}

\begin{definition}[informal] Consider the following game with a \QPT \A.
\begin{enumerate}
\item A key $k \from \KeyGen(1^n)$ is generated; $\algo A$ gets access to oracles $\Enc_k$, $\Dec_k$ and outputs circuits  $(\Samp, h, f)$;
\item Sample $m \from \Samp$; $\algo A$ receives $h(m)$, $\Enc_k(m)$, and access to an oracle for $\Enc_k$ only, and outputs a string $s$; $\algo A$ wins if $s = f(m)$.
\end{enumerate}
Then $\Pi$ is \SEMQCCA if for every \QPT \A there exists a \QPT $\algo S$ with the same winning probability but which does not get $\Enc_k(m)$ in step 2.
\end{definition}

%The following fact is straightforward.
\begin{theorem}
A classical symmetric-key encryption scheme is \INDQCCA if and only if it is \SEMQCCA.
\end{theorem}

\subsubsection{Secure constructions.}
%%%%%%%%%%%%%%%%%%%%%%%%%

Next, we show that standard pseudorandom-function-based encryption is \QCCA-secure, provided that the underlying \PRF is quantum-secure (i.e., is a \QPRF.) A \QPRF can be constructed from any quantum-secure one-way function, or directly from the \LWE assumption~\cite{Zha2012}.
Given a \PRF $f=\{f_k\}_k$, define $\PRFscheme[f]$ to be the scheme which encrypts a plaintext $m$ using randomness $r$ via $\Enc_k(m; r) = (r, f_k(r) \oplus m)$ and decrypts in the obvious way.% We show the following.

\begin{theorem}\label{thm:intro_prf_scheme}
If $f$ is a \QPRF, then $\PRFscheme[f]$ is $\INDQCCA$-secure.
\end{theorem}

We also analyze a standard permutation-based scheme. Quantum-secure \PRP{s} (i.e., \QPRP{s}) can be obtained from quantum-secure one-way functions \cite{Zha16}. Given a \PRP $P=\{P_k\}_k$, define $\PRPscheme[P]$ to be the scheme that encrypts a plaintext $m$ using randomness $r$ via $\Enc_k(m; r) = P_k(m||r)$, where $||$ denotes concatenation; to decrypt, one applies $P_k^{-1}$ and discards the randomness bits. %We show the following.

\begin{theorem}\label{thm:intro_prp_scheme}
If $P$ is a \QPRP, then $\PRPscheme[P]$ is $\INDQCCA$-secure.
\end{theorem}

We briefly describe our proof techniques for Theorems \ref{thm:intro_prf_scheme} and \ref{thm:intro_prp_scheme}. In the indistinguishability game, the adversary can use the decryption oracle prior to the challenge to (quantumly) encode information about the relevant pseudorandom function instance (i.e., $f_k$ or $P_k$) in their private, poly-sized quantum memory. From this point of view, establishing security means showing that this encoded information cannot help the adversary compute the value of the relevant function at the particular randomness used in the challenge. To prove this, we use a bound on quantum random access codes (\QRAC).
Informally, a \QRAC is a mapping from $N$-bit strings $x$ to $d$-dimensional quantum states $\rho_x$, such that given $\rho_x$, and any index $j\in [N]$, the bit $x_j$ can be recovered with some probability $p_{x,j} = \frac{1}{2} + \epsilon_{x,j}$. The average bias of such a code is the expected value of $\epsilon_{x,j}$, over uniform $x$ and $j$. A \QRAC with shared randomness further allows the encoding and decoding procedures to both depend on some random variable. 

\begin{lemma}\label{lemma:qrac}
The average bias of a quantum random access code with shared randomness that encodes $N$ bits into a $d$-dimensional quantum state is $O(\sqrt{N^{-1} \log d})$. In particular, if $N = 2^n$ and $d = 2^{\poly(n)}$ the bias is $O(2^{-n/2}\poly(n))$.
\end{lemma}

%The proof of \expref{Theorem}{thm:intro_prp_scheme} is similar to that of \expref{Theorem}{thm:intro_prf_scheme}.

\subsubsection{Key recovery against LWE.}

Our attack on \LWE encryption will make use of a new analysis of the performance of a large-modulus variant of the Bernstein-Vazirani algorithm~\cite{BV97}, in the presence of a certain type of ``rounding'' noise.

\paragraph{Quantum algorithm for linear rounding functions.}
%%%%%%%%%%%%%%%%%%%%%%%%%

In the simplest case we analyze, the oracle outputs $0$ if the inner product is small, and $1$ otherwise. Specifically, given integers $n \geq 1$ and $q \geq 2$, define a keyed family of (binary) linear rounding functions, $\Round_{\vec k,q}: \Z_q^n \longrightarrow \bit$, with key $\vec k \in \Z_q^n$, as follows:
\begin{equation*}
\Round_{\vec{k},q}(\vec x) :=
    \begin{cases}
      0 & \text{if } |\ip{\vec x, \vec k}| \leq \floor{\frac{q}{4}}, \\
      1 & \text{otherwise}.
    \end{cases}
\end{equation*}
Here $\ip{\cdot, \cdot}$ denotes the inner product modulo $q$. Our main technical contribution is the following.

\begin{theorem}[informal]\label{thm:main}
There exists a quantum algorithm which runs in time $O(n)$, makes one quantum query to $\Round_{\vec{k},q}$ (with $q \geq 2$ and unknown $\vec k \in \Z_q^{n}$), and outputs $\vec k$ with probability $4/\pi^2 - O(1/q)$.
\end{theorem}

We also show that the same algorithm succeeds against more generalized function classes, in which the oracle indicates which ``segment'' of $\Z_q$ the exact inner product belongs to. 

\paragraph{One quantum query against $\LWE$.}
%%%%%%%%%%%%%%%%%%%%%%%%%

Finally, we revisit our central question of interest: what happens to a post-quantum cryptosystem if it leaks a single quantum query? We show that, in standard \LWE-based schemes, the decryption function can (with some simple modifications) be viewed as a special case of a linear rounding function, as above. In standard symmetric-key or public-key \LWE, for instance, we decrypt a ciphertext $(\vec a,c) \in \Z_q^{n+1}$ with key $\vec k$ by outputting $0$ if $|c - \langle \vec a, \vec k \rangle| \leq \left \lfloor{\frac{q}{4}}\right \rfloor$ and $1$ otherwise. In standard \RingLWE, we decrypt a ciphertext $( u, v)$ with key $k$ (here $u, v, k$ are polynomials in $\Z_q[x]/\ip{x^n+1}$) by outputting $0$ if the constant coefficient of $ v - k \cdot u$ is small, and $1$ otherwise.

Each of these schemes is secure against adversaries with classical encryption oracle access, under the \LWE assumption. If adversaries also gain classical decryption access, then it's not hard to see that a linear number of queries is necessary and sufficient to recover the private key. Our main result is that, by contrast, only a \emph{single} quantum decryption query is required to achieve this total break. Indeed, in all three constructions described above, one can use the decryption oracle to build an associated oracle for a linear rounding function which hides the secret key. The following can then be shown using \expref{Theorem}{thm:main}.
\begin{theorem}[informal]
Let $\Pi$ be standard \LWE or standard \RingLWE encryption (either symmetric-key, or public-key.) Let $n$ be the security parameter. Then there is an efficient quantum algorithm that runs in time $O(n)$, uses one quantum query to the decryption function $\Dec_{\vec k}$ of $\Pi$,and outputs the secret key with constant probability.
\end{theorem}

\subsection{Organization}

The remainder of this paper is organized as follows. In \expref{Section}{sec:prelim}, we outline preliminary ideas that we will make use of, including cryptographic concepts, and notions from quantum algorithms. In \expref{Section}{sec:model}, we define the \QCCA model, including the two equivalent versions \INDQCCA and \SEMQCCA. In \expref{Section}{sec:constructions}, we define the \PRF and \PRP scheme, and show that they are \INDQCCA-secure. In \expref{Section}{sec:rounding}, we show how a generalization of the Bernstein-Vazirani algorithm works with probability bounded from below by a constant, even when the oracle outputs rounded values. %returns some rounded value of $\ip{\vec k,\vec x}$ (i.e.\ the oracle is a linear rounding function).
In \expref{Section}{sec:LWE}, we use the results of \expref{Section}{sec:rounding} to prove that a single quantum decryption query is enough to recover the secret key in various versions of \LWE-encryption; we also observe a similar result for a model in which the adversary can make one quantum encryption query with partial access to the randomness register.

\section{Preliminaries}\label{sec:prelim}
%%%%%%%%%%%%%%%%%%%%%%%%%

\subsection{Basic notation and conventions}
%%%%

Selecting an element $x$ uniformly at random from a finite set $X$ will be written as $x \inrand X$. If we are generating a vector or matrix with entries in $\Z_q$ by sampling each entry independently according to a distribution $\chi$ on $\Z_q$, we will write, e.g., $\vec{v} \randchi \Z_q^n$. Given a matrix $A$, $A\tp$ will denote the transpose of $A$. We will view elements $\vec v$ of $\Z_q^n$ as column vectors; the notation $\vec v\tp$ then denotes the corresponding row vector.
The notation $\negl(n)$ denotes some function of $n$ which is smaller than every inverse-polynomial. We denote the concatenation of strings $x$ and $y$ by $x||y$. We abbreviate classical probabilistic polynomial-time algorithms as \PPT algorithms.
By \textit{quantum algorithm} (or \QPT) we mean a polynomial-time uniform family of quantum circuits, where each circuit in the family is described by a sequence of unitary gates and measurements. In general, such an algorithm may receive (mixed) quantum states as inputs and produce (mixed) quantum states as outputs. Sometimes we will restrict $\QPT$s implicitly; for example, if we write $\Pr[\algo A(1^n) = 1]$ for a \QPT $\algo A$, it is implicit that we are only considering those $\QPT$s that output a single classical bit.

Every function $f: \bit^m \rightarrow \bit^\ell$ determines a unitary operator
$
U_f : \ket{x} \ket{y} \to \ket{x} \ket{y \oplus f(x)}
$
on $m + \ell$ qubits where $x \in \bit^m$ and $y \in \bit^\ell$. In this work, when we say that a quantum algorithm $\algo A$ gets (adaptive) oracle access to $f$ (written $\algo A^f$), we mean that $\algo A$ can apply the oracle unitary $U_f$. %This quantum oracle will sometimes also be denoted by $\mathcal O_f$.

Recall that a symmetric-key encryption scheme is a triple of classical probabilistic algorithms $(\KeyGen, \Enc, \Dec)$ whose run-times are polynomial in some security parameter $n$. Such a scheme must satisfy the following property: when a key $k$ is sampled by running $\KeyGen(1^n)$, then it holds that $\Dec_k (\Enc_k (m)) = m$ for all $m$ except with negligible probability in $n$. In this work, all encryption schemes will be fixed-length, i.e., the length of the message $m$ will be a fixed (at most polynomial) function of $n$.

Since the security notions we study are unachievable in the information-theoretic setting, all adversaries will be modeled by $\QPT$s. When security experiments require multiple rounds of interaction with the adversary, it is implicit that $\algo A$ is split into multiple $\QPT$s (one for each round), and that these algorithms forward their internal (quantum) state to the next algorithm in the sequence.

\subsection{Quantum-secure pseudorandomness}
%%%%

%A pseudorandom function is a family of deterministic and efficiently computable
%functions that appear random to any \PPT adversary with adaptive (classical) oracle access. Similarly, a quantum-secure pseudorandom function achieves security against \QPT adversaries with adaptive quantum oracle access. More specifically, l
Let $f: \bit^n \times \bit^m \rightarrow \bit^\ell$ be an efficiently computable function, where $n,m,\ell$ are integers and where $f$ defines a family of functions $\{f_k\}_{k\in\bit^n}$ with $f_k(x)=f(k,x)$. We say $f$ is a \textit{quantum-secure pseudorandom function} (or \QPRF) if, for every \QPT~$\algo A$,
\begin{equation}\label{eq:PRF}
\left|\Pr_{k \inrand \bit^n} \left[\algo A^{f_k}(1^n) = 1\right]
- \Pr_{g \inrand \mathcal F_m^\ell} \left[\algo A^{g}(1^n) = 1\right]\right| \leq \negl(n)\,.
\end{equation}
Here $\mathcal F_m^\ell$ denotes the set of all functions from $\bit^m$ to $\bit^\ell$. The standard method for constructing a pseudorandom function from a one-way function produces a $\QPRF$, provided that the one-way function is quantum-secure~\cite{GL89,GGM86,Zha2012}.

A quantum-secure pseudorandom permutation is a a bijective function family of quantum-secure pseudorandom functions. More specifically, consider a function $P: \bit^n \times \bit^m \rightarrow \bit^m$, where $n$ and $m$ are integers, such that each function $P_k(x)=P(k,x)$ in the corresponding family $\{P_k\}_{k\in\bit^n}$ is bijective. We say $P$ is a \textit{quantum-secure pseudorandom permutation} (or \QPRP) if, for every \QPT $\algo A$ with access to both the function and its inverse,
\begin{equation}\label{eq:PRP}
\left|\Pr_{k \inrand \bit^n} \left[\algo A^{P_k,P_k^{-1}}(1^n) = 1\right]
- \Pr_{\pi \inrand \mathcal{P}_m} \left[\algo A^{\pi,\pi^{-1}}(1^n) = 1\right]\right| \leq \negl(n)\,,
\end{equation}
where $\mathcal P_m$ denotes the set of permutations over $m$-bit strings. %Throughout this work,we shall assume strong \QPRP{s} under the above definition, i.e. such that the security is maintained despite additional access to an inverse.
One can construct \QPRP{s} from quantum-secure one-way functions~\cite{Zha16}.

\subsection{Quantum random access codes}
%%%

A \emph{quantum random access code} (\QRAC) is a two-party scheme for the following scenario involving two parties Alice and Bob~\cite{Nayak99}:
\begin{enumerate}
  \item Alice gets $x\in\{0,1\}^N$ and encodes it as a $d$-dimensional quantum state~$\rho_x$.
  \item Bob receives $\rho_x$ from Alice, and some index $i\in\{1,\dots,N\}$, and is asked to recover the $i$-th bit of $x$, by performing some measurement on $\rho_x$.
  \item They win if Bob's output agrees with $x_i$ and lose otherwise.
\end{enumerate}
We can view a \QRAC scheme as a pair of (not necessarily efficient) quantum algorithms: one for encoding, and another for decoding. We remark that the definition of a \QRAC does not require a bound on the number of qubits; the interesting question is with what parameters a \QRAC can actually exist.

A variation of the above scenario allows Alice and Bob to use \emph{shared randomness} in their encoding and decoding operations \cite{ALMO08}. %(note that shared randomness \textit{per se} does not allow them to communicate). 
Hence,  Alice and Bob can pursue probabilistic strategies with access to the same random variable.
% Finally, the definiton also takes into account
%that the encoder and decoder are allowed to share an entangled state as part of $\rho_x$.

Define the average bias of a \QRAC with shared randomness as $\epsilon = p_\text{win} - 1/2$, where $p_\text{win}$ is the winning probability averaged over $x \inrand \bit^N$ and $i \inrand \{1, \dotsc, N\}$.

\subsection{Quantum Fourier transform}
%%%

For any positive integer $q$, the quantum Fourier transform over $\Z_q$ is defined by the operation
$$
\QFT_{\Z_q}\ket{x} = \frac{1}{\sqrt{q}} \sum_{y \in \Z_q} \omega_q^{x\cdot y} \ket{y},
$$
where $\omega_q = e^{ \frac{2\pi i}{q}}$.
Due to early work by Kitaev \cite{Kitaev95}, this variant of the Fourier transform can be implemented using quantum phase estimation in complexity polynomial in $\log q$.
An improved approximate implementation of this operation is due to Hales and Hallgren \cite{HH00}.

%%%%%%%%%%%%%%%%%%%%%%%%%
\section{The QCCA1 security model}\label{sec:model}
%%%%%%%%%%%%%%%%%%%%%%%%%

\subsection{Quantum oracles}
%%%%%%%%%%%%%%%%%%%%%%%%%

In our setting, adversaries will (at various times) have quantum oracle access to the classical functions $\Enc_k$ and $\Dec_k$.
The case of the deterministic decryption function $\Dec_k$ is simple: the adversary gets access to the unitary operator $U_{\Dec_k} : \ket{c} \ket{m} \mapsto \ket{c} \ket{m \oplus \Dec_k(c)}.$ For encryption, to satisfy \INDCPA security, $\Enc_k$ must be probabilistic and thus does not correspond to any single unitary operator. Instead, each encryption oracle call of the adversary will be answered by applying a unitary sampled uniformly from the family $\{U_{\Enc_k, r}\}_r$ where
$$
U_{\Enc_k, r} : \ket{m} \ket{c} \mapsto \ket{m} \ket{c \oplus \Enc_k(m ; r)}
$$
and $r$ varies over all the possible values of the randomness register of $\Enc_k$. Note that, since $\Enc_k$ and $\Dec_k$ are required to be probabilistic polynomial-time algorithms provided by the underlying classical symmetric-key encryption scheme, both $U_{\Enc_k, r}$ and $U_{\Dec_k}$ correspond to efficient and reversible quantum operations. For the sake of brevity, we adopt the convenient notation $\Enc_k$ and $\Dec_k$ to refer to the above quantum oracles for encryption and decryption respectively.

\subsection{Ciphertext indistinguishability}
%%%%%%%%%%%%%%%%%%%%%%%%%

We now define indistinguishability of encryptions (for classical, symmetric-key schemes) against non-adaptive quantum chosen-ciphertext attacks.
\begin{definition}[\INDQCCA] \label{def:ind_qcca}
Let $\Pi = (\KeyGen, \Enc, \Dec)$ be an encryption scheme, $\algo A$ a \QPT, and $n$ the security parameter. Define $\IndGame(\Pi, \algo A, n)$ as follows.
\begin{enumerate}
\item \phase{Setup} A key $k \from \KeyGen(1^n)$ and a bit $b \inrand \bit$ are generated;
\item \phase{Pre-challenge} $\algo A$ gets access to oracles $\Enc_k$ and $\Dec_k$, and outputs $(m_0, m_1)$;
\item \phase{Challenge} $\algo A$ gets $\Enc_k(m_b)$ and access to $\Enc_k$ only, and outputs a bit~$b'$;
\item \phase{Resolution} $\algo A$ wins if $b = b'$.
\end{enumerate}
Then $\Pi$ has indistinguishable encryptions under non-adaptive quantum chosen ciphertext attack (or is \INDQCCA) if,
for every \QPT $\algo A$,
$$
\Pr[\algo A \text{ wins } \IndGame(\Pi, \algo A, n)] \leq 1/2 + \negl(n)\,.
$$
\end{definition}

By inspection, one immediately sees that our definition lies between the established notions of \INDQCPA and \INDQCCAA~\cite{BJ15,GHS16,BZ13}. %We will later show that these separations are strict. \snote{Where do we show that?}
It will later be convenient to work with a variant of the game \IndGame, which we now define.

\begin{definition}[\IndGameR]\label{def:ind_qcca_r}
We define the experiment $\IndGameR(\Pi, \algo A, n)$ just as $\IndGame(\Pi,\algo A,n)$, except that in the pre-challenge phase $\algo A$ only outputs a single message $m$, and in the challenge phase $\algo A$ receives $\Enc_k(m)$ if $b=0$, and $\Enc_k(x)$ for a uniformly random message $x$ if $b=1$.
\end{definition}

Working with \IndGameR rather than \IndGame does not change security. Specifically (as we show in \expref{Appendix}{sec:cca1-equiv}), $\Pi$ is \INDQCCA if and only if, for every \QPT $\algo A$,
$
\Pr[\algo A \text{ wins }\IndGameR(\Pi, \algo A, n)] \leq 1/2 + \negl(n)\,.
$

\subsection{Semantic security}
%%%%%%%%%%%%%%%%%%%%%%%%%

In semantic security, rather than choosing a pair of challenge plaintexts, the adversary chooses a \emph{challenge template}: a triple of circuits $(\Samp, h, f)$, where $\Samp$ outputs plaintexts from some distribution $\mathcal D_\Samp$, and $h$ and $f$ are functions with domain the support of $\mathcal D_\Samp$. The intuition is that $\Samp$ is a distribution of plaintexts $m$ for which the adversary, if given information $h(m)$ about $m$ together with an encryption of $m$, can produce some new information~$f(m)$.

\begin{definition}[\SEMQCCA]\label{def:sem_qcca} Let $\Pi = (\KeyGen, \Enc, \Dec)$ be an encryption scheme, and consider the experiment $\SemGame(b)$ (with parameter $b \in \{\real, \simul\}$) with a \QPT $\algo A$, defined as follows.
\begin{enumerate}
\item \phase{Setup} A key $k \from \KeyGen(1^n)$ is generated;
\item \phase{Pre-challenge} $\algo A$ gets access to oracles $\Enc_k$ and $\Dec_k$, and outputs a challenge template $(\Samp, h, f)$;
\item \phase{Challenge} A plaintext $m \inrand \Samp$ is generated; $\algo A$ receives $h(m)$ and gets access to an oracle for $\Enc_k$ only; if $b = \real$, $\algo A$ also receives $\Enc_k(m)$; $\algo A$ outputs a string $s$;
\item \phase{Resolution} $\algo A$ wins if $s = f(m)$.
\end{enumerate}
$\Pi$ has semantic security under non-adaptive quantum chosen ciphertext attack (or is \SEMQCCA) if, for every \QPT $\algo A$, there exists a \QPT $\algo S$ such that the challenge templates output by $\algo A$ and $\algo S$ are identically distributed, and
$$
\bigl|\Pr[\algo A \text{ wins } \SemGame(\real)] - \Pr[\algo S \text{ wins } \SemGame(\simul)]\bigr| \leq \negl(n)\,.
$$
\end{definition}

Our definition is a straightforward modification of \SEMQCPA~\cite{GHS16,BZ13}; the modification is to give $\algo A$ and $\algo S$ oracle access to $\Dec_k$ in the pre-challenge phase.

\begin{theorem}
Let $\Pi=(\KeyGen,\Enc, \Dec)$ be a symmetric-key encryption scheme. Then, $\Pi$ is $\INDQCCA$-secure if and only if $\Pi$ is $\SEMQCCA$-secure.
\end{theorem}

The classical proof of the above (see, e.g.,~\cite{Goldreich2004}) carries over directly to the quantum case. This was already observed for the case of $\QCPA$ by~\cite{GHS16}, and extends straightforwardly to the case where both the adversary and the simulator gain oracle access to $\Dec_k$ in the pre-challenge phase.\footnote{In fact, the proof works even if $\Dec_k$ access is maintained during the challenge, so the result is really that \INDQCCAA is equivalent to \SEMQCCAA.}

%%%%%%%%%%%%%%%%%%%%%
\section{Secure Constructions}\label{sec:constructions}
%%%%%%%%%%%%%%%%%%%%%

\subsection{\PRF scheme}
%%%%%%%%%%%%%%%%%%%%%

Let us first recall the standard symmetric-key encryption based on pseudorandom functions.

\begin{construction}[\PRF scheme]\label{cons:prf}
Let $n$ be the security parameter and let $f:\bit^n \times \bit^n \longrightarrow \bit^n$ be an efficient family of functions $\{f_k\}_k$. Then, the symmetric-key encryption scheme $\PRFscheme[f]=(\KeyGen,\Enc,\Dec)$ is defined as follows:
\begin{enumerate}
%\item \phase{\textbf{Key generation}} generate a key $\vec k \rand \bit^n$;
%\item \phase{\textbf{Encryption}} on message $\vec m$, choose a random string $\vec r \rand \bit^n$ and output $(\vec r,f_k(\vec r) \oplus \vec m)$;
%\item \phase{\textbf{Decryption}} on ciphertext $(\vec r,\vec c)$, output $\vec c \oplus f_k(\vec r)$;
\item \phase{\KeyGen} output $k \rand \bit^n$;
\item \phase{\Enc} to encrypt $m \in \{0,1\}^n$, choose $r \rand \bit^n$ and output $(r, f_k(r) \oplus m)$;
\item \phase{\Dec} to decrypt $(r, c) \in \{0,1\}^n \times \{0,1\}^n$, output $c \oplus f_k(r)$;
\end{enumerate}
\end{construction}

For simplicity, we chose a particularly simple set of parameters for the \PRF, so that key length, input size, and output size are all equal to the security parameter. It is straightforward to check that the definition (and our results below) are valid for arbitrary polynomial-size parameter choices.

We show that the above scheme satisfies \QCCA, provided that the underlying \PRF is secure against quantum queries.

\begin{theorem}\label{th:prf_scheme}
If $f$ is a \QPRF, then $\PRFscheme[f]$ is $\INDQCCA$-secure.
\end{theorem}

\begin{proof}
Fix a \QPT adversary \A against $\Pi := \PRFscheme[f] = (\KeyGen, \Enc, \Dec)$ and let $n$ denote the security parameter. It will be convenient to split \A into the pre-challenge algorithm $\algo A_1$ and the challenge algorithm $\algo A_2$.

We will work with the single-message variant of \IndGame, \IndGameR, described below as \textsc{Game~0}. In \expref{Appendix}{sec:cca1-equiv}, we show that $\Pi$ is \INDQCCA if and only if no \QPT adversary can win \IndGameR with non-negligible bias. We first show that a version of \IndGameR where we replace $f$ with a random function, called \textsc{Game 1} below, is indistinguishable from \IndGameR, so that the winning probabilities cannot differ by a non-negligible amount. We then prove that no adversary can win \textsc{Game 1} with non-negligible bias by showing how any adversary for \textsc{Game 1} can be used to make a quantum random access code with the same bias.

\begin{figure}[ht]
\centering
\begin{tikzpicture}[> = latex, scale = 1.2]

\node at (-1.55,.5) {$1^n$};
\draw[->] (-1.3,.5)--(-.8,.5);
\draw (-.8,0) rectangle (1,1);
\node at (.1,.5) {$\algo A_1$};

\node at (2,.9) {$\ket{\psi}$};
\draw[->] (1,.75)--(3,.75);

\node at (1.375,.4) {$m^*$};
\draw[->] (1,.25)--(1.75,.25);

\node at (2,.25) {$\Phi_b$};
\draw (1.75,0) rectangle (2.25,.5);

\node at (2.625,.4) {$c^*$};
\draw[->] (2.25,.25) -- (3,.25);

\draw (3,0) rectangle (4,1);
\node at (3.5,.5) {$\algo A_2$};
\draw[->] (4,.5) -- (4.5,.5);
\node at (4.75,.5) {$b'$};

\draw[->] (-.6,0) -- (-.6,-.3);
\draw[->] (-.2,-.3) -- (-.2,0);
\draw (-.75,-.7) rectangle (-.05,-.3);
\node at (-.4,-.5) {$\Enc_k$};

\draw[->] (.4,0) -- (.4,-.3);
\draw[->] (.8,-.3) -- (.8,0);
\draw (.25,-.7) rectangle (.95,-.3);
\node at (.6,-.5) {$\Dec_k$};

\draw[->] (3.3,0) -- (3.3,-.3);
\draw[->] (3.7,-.3) -- (3.7,0);
\draw (3.15,-.7) rectangle (3.85,-.3);
\node at (3.5,-.5) {$\Enc_k$};

\end{tikzpicture}
\caption{\IndGameR from \expref{Definition}{def:ind_qcca_r}.}\label{fig:indgame}
\end{figure}

\begin{description}
\item[Game 0:] This is the game $\IndGameR(\Pi, \algo A, n)$, which we briefly review for convenience (see also \expref{Figure}{fig:indgame}). In the pre-challenge phase, $\algo A_1$ gets access to oracles $\Enc_k$ and $\Dec_k$, and outputs a message $m^*$ while keeping a private state $\ket{\psi}$ for the challenge phase. In the challenge phase, a random bit $b \inrand \bit$ is sampled, and $\algo A_2$ is run on input $\ket{\psi}$ and a challenge ciphertext
$$
c^* := \Phi_b(m^*) :=
\begin{cases}
  \Enc_k(m^*) & \mbox{if } b = 0, \\
  \Enc_k(x)   & \mbox{if } b = 1.
\end{cases}
$$
Here $\Enc_k(x) := (r^*, f_k(r^*)\oplus x)$ where $r^*$ and $x$ are sampled uniformly at random. In the challenge phase, $\algo A_2$ only has access to $\Enc_k$ and must output a bit $b'$. $\algo A$ wins if $\delta_{bb'}=1$, so we call $\delta_{bb'}$ the outcome of the game.

\item[Game 1:] This is the same game as \textsc{Game 0}, except we replace $f_k$ with a uniformly random function $F:\bit^n\rightarrow\bit^n$.

%\item[\textsc{Game 2}:] This is the same game as \textsc{Game 1}, except the challenge phase is modified as follows. In place of a normally generated challenge ciphertext, $\algo A_2$ receives
%$$
%c^* := (r, f_k(r) \oplus s \oplus m_b)
%$$
%where $s$ is sampled uniformly at random, and $r$ and $b$ are sampled as in \textsc{Game 1}.
\end{description}

%We first observe that $s$ is a uniformly random string that is independent of all other random variables in \textsc{Game 2}. By the information-theoretic security of the one-time pad, it follows that the success probability of \A in \textsc{Game 2} is at most $1/2$.

First, we show that for any adversary $\algo A$, the outcome when $\algo A$ plays \textsc{Game 0} is at most negligibly different from the outcome when $\algo A$ plays \textsc{Game 1}. We do this by constructing a quantum oracle distinguisher $\algo D$ that distinguishes between the \QPRF $\{f_k\}_k$ and a true random function, with distinguishing advantage
$$
\bigl|\Pr[1 \from \textsc{Game 0}] - \Pr[1 \from \textsc{Game 1}]\bigr|,
$$
which must then be negligible since $f$ is a \QPRF. The distinguisher $\algo D$ gets quantum oracle access to a function $g$, which is either $f_k$, for a random $k$, or a random function, and proceeds by simulating $\algo A$ playing \IndGameR as follows:
\begin{enumerate}
\item Run $\algo A_1$, answering encryption queries using classical calls to $g$ in place of $f_k$, and answering decryption queries using quantum oracle calls to $g$:
$$
\ket{r}\ket{c}\ket{m}
\mapsto \ket{r}\ket{c}\ket{m \oplus c}
\mapsto \ket{r}\ket{c}\ket{m \oplus c \oplus g(r)}\,;
$$
\item Simulate the challenge phase by sampling $b \inrand \bit$ and encrypting the challenge using $g$ in place of $f_k$;
      run $\algo A_2$ and simulate encryption queries as before;
\item When $\algo A_2$ outputs $b'$, output $\delta_{bb'}$.
\end{enumerate}

It remains to show that no \QPT adversary can win \textsc{Game 1} with non-negligible probability. To do this, we design a quantum random access code from any adversary, and use the lower bound on the bias given in \expref{Lemma}{lemma:qrac}.

\paragraph{Intuition.} We first give some intuition. In an encryption query, the adversary, either $\algo A_1$ or $\algo A_2$, queries a message, or a superposition of messages $\sum_m\ket{m}$, and gets back $\sum_m\ket{m}\ket{r,m\oplus F(r)}$ for a random $r$, from which he can easily get a sample $(r,F(r))$. Thus, in essence, an encryption query is just classically sampling a random point of $F$.

In a decryption query, which is only available to $\algo A_1$, the adversary sends a ciphertext, or a superposition of ciphertexts, $\sum_{r,c}\ket{r,c}$ and gets back $\sum_{r,c}\ket{r,c}\ket{c\oplus F(r)}$, from which he can learn $\sum_r\ket{r,F(r)}$. Thus, a decryption query allows $\algo A_1$ to query $F$, in superposition. Later in the challenge phase, $\algo A_2$ gets an encryption $(r^*,m\oplus F(r^*))$ and must decide if $m=m^*$. Since $\algo A_2$ no longer has access to the decryption oracle, which allows him to query $F$, there seem to be two possible ways $\algo A_2$ could learn $F(r^*)$:
\begin{enumerate}
\item $\algo A_2$ gets lucky in one of his at most $\poly(n)$ many queries to $\Enc_k$ and happens to sample $(r^*,F(r^*))$;
\item Or, the adversary is somehow able to use what he learned while he had access to $\Dec_k$, and thus $F$, to learn $F(r^*)$, meaning that the $\poly(n)$-sized quantum memory $\algo A_1$ sends to $\algo A_2$, that can depend on queries to $F$, but which cannot depend on $r^*$, allows $\algo A_2$ to learn $F(r^*)$.
\end{enumerate}
The first possibility is exponentially unlikely, since there are $2^n$ possibilities for $r^*$. As we will see shortly, the second possibility would imply a very strong quantum random access code. It would essentially allow $\algo A_1$ to interact with $F$, which contains $2^n$ values, and make a state, which must necessarily be of polynomial size, such that $\algo A_2$ can use that state to recover $F(r^*)$ for any of the $2^n$ possible values of $r^*$, with high probability. We now formalize this intuition. To clarify notation, we will use boldface to denote the shared randomness bitstrings.

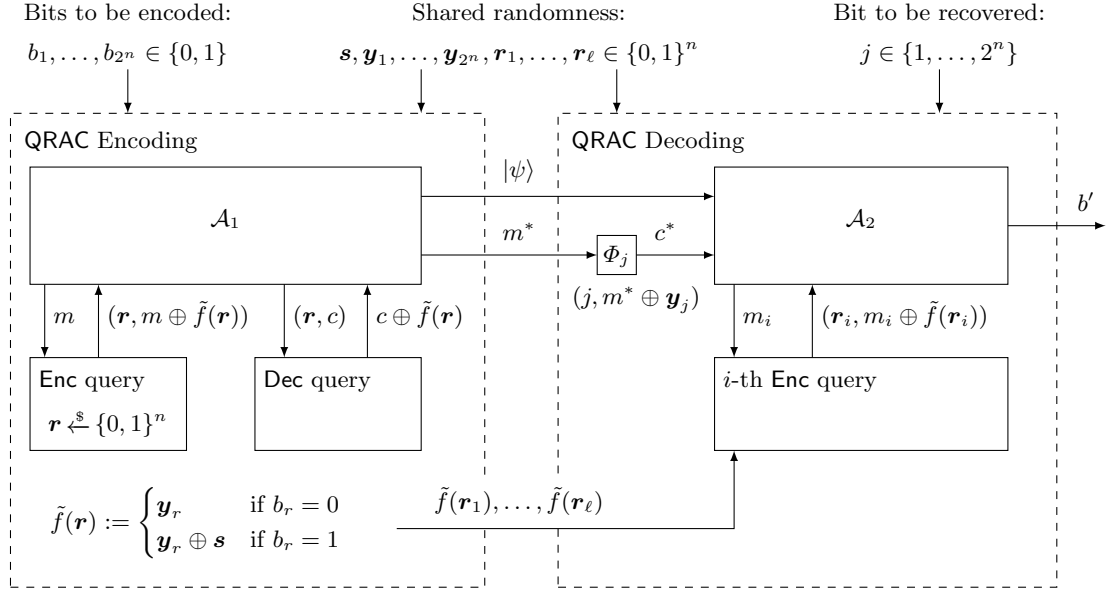
\begin{figure}[ht]
\centering
\begin{tikzpicture}[> = latex, scale = 1.3,
  l/.style = {anchor = base west},
  every node/.style = {anchor = base}
]

% Inputs

\draw[->] (1,2.2)--(1,1.75);
\node at (1,2.7) {Bits to be encoded:};
\node at (1,2.3) {$b_1, \dotsc, b_{2^n} \in \bit$};

\draw[->] (9.3,2.2)--(9.3,1.75);
\node at (9.3,2.7) {Bit to be recovered:};
\node at (9.3,2.3) {$j \in \{1, \dotsc, 2^n\}$};

\draw[->] (4,2.2)--(4,1.75);
\draw[->] (6,2.2)--(6,1.75);
\node at (5,2.7) {Shared randomness:};
\node at (5,2.3) {$\vec s, \vec y_1, \dotsc, \vec y_{2^n}, \vec r_1, \dotsc, \vec r_\ell \in \bit^n$};

% QRAC boxes

\draw[dashed] (-.2,1.75) rectangle (4.65,-3.1);
\node[l] at (-.15,1.4) {\QRAC Encoding};

\draw[dashed] (5.4,1.75) rectangle (10.5,-3.1);
\node[l] at (5.45,1.4) {\QRAC Decoding};

% A1 and A2 boxes

\draw (0,0) rectangle (4,1.2);
\node at (2,.6) {$\algo A_1$};

\draw (7,0) rectangle (10,1.2);
\node at (8.5,.6) {$\algo A_2$};

\draw[->] (10,.6)--(11,.6);
\node at (10.8,.75) {$b'$};

% Enc query

\draw[->] (.15,0)--(.15,.-.75);
\node[l] at (.15,-.4) {$m$};

\draw[<-] (.7,0)--(.7,.-.75);
\node[l] at (.7,-.4) {$(\vec r, m \oplus \tilde{f}(\vec r))$};

\draw (0,-1.7) rectangle (1.6,-.75);
\node[l] at (0,-1.05) {$\Enc$ query};
\node at (.8,-1.5) {$\vec r \inrand \bit^n$};

% Dec query

\draw[->] (2.6,0)--(2.6,.-.75);
\node[l] at (2.6,-.4) {$(\vec r, c)$};

\draw[<-] (3.45,0)--(3.45,.-.75);
\node[l] at (3.45,-.4) {$c \oplus \tilde{f}(\vec r)$};

\draw (2.3,-1.7) rectangle (4.0,-.75);
\node[l] at (2.3,-1.05) {$\Dec$ query};

% Enc_i query

\draw[->] (7.2,0)--(7.2,-.75);
\node[l] at (7.2,-.4) {$m_i$};

\draw[<-] (8.0,0)--(8.0,-.75);
\node[l] at (8.0,-.4) {$(\vec r_i,m_i\oplus \tilde{f}(\vec r_i))$};

\draw (7, -1.7) rectangle (10,-.75);
\node[l] at (7,-1.05) {$i$-th $\Enc$ query};

% QRAC message

\draw[->] (3.75,-2.5)--(7.2,-2.5)--(7.2,-1.7);
\node[fill = white, inner sep = 4pt] at (5,-2.3) {$\tilde{f}(\vec r_1), \dotsc, \tilde{f}(\vec r_\ell)$};

\draw[->] (4,.9)--(7,.9);
\node at (5,1.1) {$\ket{\psi}$};

\draw[->] (4,.3)--(5.8,.3);
\node at (5,.45) {$m^*$};

% Challenge

\draw (5.8,.1) rectangle (6.2,.5);
\node at (6.0,.23) {$\Phi_j$};

\draw[->] (6.2,.3)--(7,.3);
\node at (6.5, 0.45) {$c^*$};
\node at (6.2,-0.20) {$(j,m^* \oplus \vec y_j)$};

% PRF function

\node at (1.7,-2.5) {$\tilde{f}(\vec r) := \begin{cases}
\vec y_r & \mbox{if } b_r = 0 \\
\vec y_r \oplus \vec s & \mbox{if } b_r = 1
\end{cases}$};

\end{tikzpicture}
\caption{Quantum random access code construction for the \PRF scheme.}\label{fig:code}
\end{figure}

\paragraph{Construction of a quantum random access code.} Let $\algo A$ be a \QPT adversary with winning probability $p$. Let $\ell=\poly(n)$ be an upper bound on the number of queries made by $\algo A_2$. Recall that a random access code consists of an encoding procedure that takes (in this case) $2^n$ bits $b_1,\dots,b_{2^n}$, and outputs a state $\rho$ of dimension (in this case) $2^{\poly(n)}$, such that a decoding procedure, given $\rho$ and an index $j\in\{1,\dots,2^n\}$ outputs $b_j$ with some success probability.
We define a quantum random access code as follows (see also \expref{Figure}{fig:code}).

\begin{description}
\item[Encoding.] Let $b_1,\dots,b_{2^n}\in \bit$ be the string to be encoded. Let $\vec s,\vec y_1,\dots, \vec y_{2^n}\in \bit^n$ be given by the first $n(1+2^n)$ bits of the shared randomness, and let $\vec r_1,\dots,\vec r_\ell\in\bit^n$ be the next $\ell n$ bits. Define a function $\tilde{f}:\bit^n\rightarrow \bit^n$ as follows. For $\vec r\in\bit^n$, we will slightly abuse notation by letting $r$ denote the corresponding integer value between 1 and $2^n$. Define $\tilde{f}(\vec r)=\vec y_r\oplus b_r\vec s$. Run $\algo A_1$, answering encryption and decryption queries using $\tilde{f}$ in place of $F$. Let $m^*$ and $\ket{\psi}$ be the outputs of $\algo A_1$ (see \expref{Figure}{fig:indgame}). Output $\rho=(\ket{\psi},m^*,\tilde{f}(\vec r_1),\dotsc,\tilde{f}(\vec r_\ell))$.
\item[Decoding.] Let $j\in\{1,\dots,2^n\}$ be the index of the bit to be decoded (so given $\rho$ as above, the goal is to recover $b_j$). Decoding will make use of the values $\vec s,\vec y_1,\dots,\vec y_{2^n},\vec r_1,\dots,\vec r_\ell$ given by the shared randomness. Upon receiving a query $j\in \{1,\dots,2^n\}$, run $\algo A_2$ with inputs $\ket{\psi}$ and $(j,m^*\oplus \vec y_j)$. On $\algo A_2$'s $i$-th encryption oracle call, use randomness $\vec r_i$, so that if the input to the oracle is $\ket{m,c}$, the state returned is $\ket{m,c\oplus (\vec r_i,m\oplus \tilde{f}({\vec r_i}))}$ (note that $\tilde{f}({\vec r_i})$ is given as part of $\rho$).
Return the bit $b'$ output by $\algo A_2$.
\end{description}

\paragraph{Average bias of the code.} We claim that the average probability of decoding correctly, taken over all choices of $b_1,\dots,b_{2^n}\in\bit$ and $j\in\{1,\dots,2^n\}$, is exactly $p$, the success probability of $\algo A$. To see this,
first note that from $\algo A$'s perspective, this is exactly \textsc{Game 1}: the function $\tilde{f}$ is a uniformly random function, and the queries are responded to just as in \textsc{Game 1}.
Further, note that if $b_j=0$, then $m^*\oplus \vec y_j=m^* \oplus \tilde{f}(j)$, so the correct guess for $\algo A_2$ would be $0$, and if $b_j=1$, then $m^* \oplus \vec y_j = m^* \oplus \tilde{f}(j)\oplus \vec s = \vec x \oplus\tilde{f}(j)$ for the uniformly random string $\vec x =m^*\oplus \vec s$, so the correct guess for $\algo A_2$ would be 1.

Therefore, the average bias of the code is $p-1/2$. We also observe that $\rho$ has dimension at most $2^{\poly(n)}$, since $\ket{\psi}$ must be a $\poly(n)$-qubit state ($\algo A_1$ only runs for $\poly(n)$ time), and $\ell$, the number of queries made by $\algo A_2$ must be $\poly(n)$, since $\algo A_2$ only runs for $\poly(n)$ time.
As this code encodes $2^n$ bits into a state of dimension $2^{\poly(n)}$, by \expref{Lemma}{lemma:qrac} (proven in \expref{Appendix}{sec:qrac}), the bias is $O(2^{-n/2}\poly(n))=\negl(n)$, so $p\leq \frac{1}{2}+\negl(n)$.\qed
\end{proof}

\subsection{\PRP scheme}
%%%%%%%%%%%%%%%%%%%%%

We now prove the \INDQCCA security of a standard encryption scheme based on pseudorandom permutations. 

\begin{construction}[\PRP scheme]\label{cons:prp}
Let $n$ be the security parameter and let $P:\bit^n \times \bit^{2n} \longrightarrow \bit^{2n}$ be an efficient family of permutations $\{P_k\}_k$. Then, the symmetric-key encryption scheme $\PRPscheme[f]=(\KeyGen,\Enc,\Dec)$ is defined as follows:
\begin{enumerate}
\item \phase{\KeyGen} output $k \inrand \bit^n$;
\item \phase{\Enc} to encrypt $m \in \bit^n$, choose $r \rand\bit^n$ and output $P_k(m || r)$;
\item \phase{\Dec} to decrypt $c \in \bit^{2n}$, output the first $n$ bits of $P_k^{-1}(c)$.
\end{enumerate}
\end{construction}

As before, we chose a simple set of parameters; in general, the randomness length, plaintext length, and security parameter can be related by arbitrary polynomials.

\begin{theorem}\label{th:prp_scheme}
If $P$ is a \QPRP, then $\PRPscheme[P]$ is $\INDQCCA$-secure.
\end{theorem}

\begin{proof}
We follow a similar proof strategy as with the \PRF scheme.
Fix a \QPT adversary \A against $\Pi := \PRPscheme[P] = (\KeyGen, \Enc, \Dec)$ and let $n$ denote the security parameter.
We have that $\Pi$ is \INDQCCA if and only if no \QPT adversary can win \IndGameR with non-negligible bias.
First, we show that a version of \IndGameR where we replace $P$ with a random permutation, described below as \textsc{Game 1}, is indistinguishable from \IndGameR, so that the winning probabilities cannot differ by a non-negligible amount.
We then prove that no adversary can win \textsc{Game 1} with non-negligible bias, by showing how any adversary for \textsc{Game 1} can be used to make a quantum random access code with the same bias.

\begin{description}
\item[Game 0:] In the pre-challenge phase, $\algo A_1$ gets access to oracles $\Enc_k$ and $\Dec_k$.
In the challenge phase, $\algo A_1$ outputs $m$ and its private data $\ket{\psi}$; a random bit $b \inrand \bit$ is sampled, and $\algo A_2$ is run on input $\ket{\psi}$ and a challenge ciphertext
$$
c^* := \begin{cases}
\Enc_k(m^*) = P_k(m^*|| r^*) & \mbox{if } b = 0 ,\\
\Enc_k(x)   = P_k(x  || r^*) & \mbox{if } b = 1,
\end{cases}
$$
where $r^* \rand \bit^n$ and $x$ is sampled uniformly at random. In the challenge phase, $\algo A_2$ has oracle access to $\Enc_k$ only and outputs a bit $b'$.
The outcome of the game is simply the bit~$\delta_{bb'}$.

\item[Game 1:] This is the same game as \textsc{Game 0}, except we now replace $P_k$ with a perfectly random permutation $\pi:\bit^{2n}\rightarrow\bit^{2n}$.
\end{description}

We show that for any adversary $\algo A$, the outcome when $\algo A$ plays \textsc{Game 0} is at most negligibly different from the outcome when $\algo A$ plays \textsc{Game 1}.
We construct a quantum oracle distinguisher $\algo D$ that distinguishes between $P_k$ and a perfectly random permutation, with distinguishing advantage
$$
\left|\Pr[1 \from \textsc{Game 0}] - \Pr[1 \from \textsc{Game 1}]\right|,
$$
which must then be negligible since $P_k$ is a \QPRP. Here, the distinguisher $\algo D$ receives quantum oracle access to a function $\varphi$, which is either $P_k$ for a random $k$, or a random permutation $\pi$, and proceeds by simulating $\algo A$ playing \IndGameR as follows:
\begin{enumerate}
\item Run $\algo A_1$, answering encryption queries using oracle calls to $\varphi$ in place of $P_k$,
where for a given input and via randomness $r$,
$$\Enc: \, \ket{m}\ket{c}
\mapsto \ket{m}\ket{c \oplus \varphi(m||r)}.
$$
Answer decryption queries using quantum oracle calls to $\tilde \varphi^{-1}$, a function that first computes $\varphi^{-1}$ but then (analogous to the \PRP construction) discards
the last $n$ bits of the pre-image corresponding to the randomness, i.e.
$$
\Dec: \, \ket{c}\ket{m} \mapsto \ket{c}\ket{m \oplus \tilde \varphi^{-1}(c)}.
$$
\item Simulate the challenge phase by sampling $b \inrand \bit$ and encrypting using a randomness $r^*$ together with a classical call to $\varphi$ in place of $P_k$;
run $\algo A_2$ and simulate encryption queries as before.
\item When $\algo A_2$ outputs $b'$, output $\delta_{bb'}$.
\end{enumerate}

It remains to show that no \QPT adversary can win \textsc{Game 1} with non-negligible probability.
To do this, we will again design a random access code from any adversary's strategy with success probability $p$,
and use the lower bound on the bias given in \expref{Lemma}{lemma:qrac}.
We will then construct a \QRAC with bias $\negl(n)$ from this adversary, and hence conclude that $p\leq \frac{1}{2}+\negl(n)$.

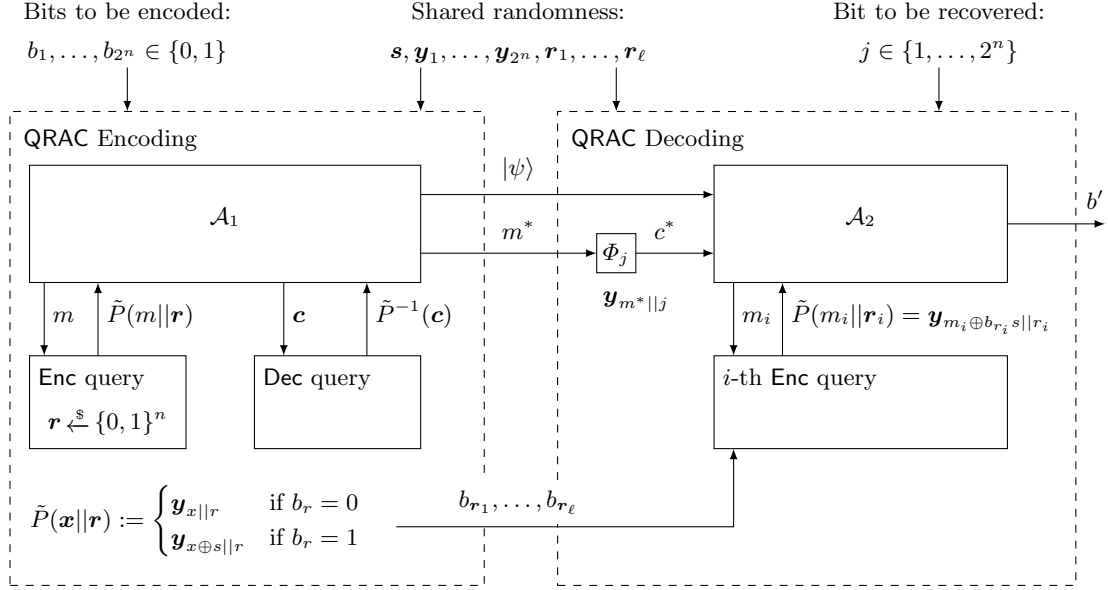
\begin{figure}[t]
\centering
\begin{tikzpicture}[> = latex, scale = 1.3,
  l/.style = {anchor = base west},
  every node/.style = {anchor = base}
]

% Inputs

\draw[->] (1,2.2)--(1,1.75);
\node at (1,2.7) {Bits to be encoded:};
\node at (1,2.3) {$b_1, \dotsc, b_{2^n} \in \bit$};

\draw[->] (9.3,2.2)--(9.3,1.75);
\node at (9.3,2.7) {Bit to be recovered:};
\node at (9.3,2.3) {$j \in \{1, \dotsc, 2^n\}$};

\draw[->] (4,2.2)--(4,1.75);
\draw[->] (6,2.2)--(6,1.75);
\node at (5,2.7) {Shared randomness:};
\node at (5,2.3) {$\vec s,\vec y_1,\dots,\vec y_{2^n},\vec r_1,\dots, \vec r_{\ell}$};

% QRAC boxes

\draw[dashed] (-.2,1.75) rectangle (4.65,-3.1);
\node[l] at (-.15,1.4) {\QRAC Encoding};

\draw[dashed] (5.4,1.75) rectangle (10.7,-3.1);
\node[l] at (5.45,1.4) {\QRAC Decoding};

% A1 and A2 boxes

\draw (0,0) rectangle (4,1.2);
\node at (2,.6) {$\algo A_1$};

\draw (7,0) rectangle (10,1.2);
\node at (8.5,.6) {$\algo A_2$};

\draw[->] (10,.6)--(11,.6);
\node at (10.9,.75) {$b'$};

% Enc query

\draw[->] (.15,0)--(.15,.-.75);
\node[l] at (.15,-.4) {$m$};

\draw[<-] (.7,0)--(.7,.-.75);
\node[l] at (.7,-.4) {$\tilde P(m || \vec r)$};

\draw (0,-1.7) rectangle (1.6,-.75);
\node[l] at (0,-1.05) {$\Enc$ query};
\node at (.8,-1.5) {$\vec r \inrand \bit^n$};

% Dec query

\draw[->] (2.6,0)--(2.6,.-.75);
\node[l] at (2.6,-.4) {$\vec c$};

\draw[<-] (3.45,0)--(3.45,.-.75);
\node[l] at (3.45,-.4) {$\tilde P^{-1}(\vec c)$};

\draw (2.3,-1.7) rectangle (4.0,-.75);
\node[l] at (2.3,-1.05) {$\Dec$ query};

% Enc_i query

\draw[->] (7.2,0)--(7.2,-.75);
\node[l] at (7.2,-.4) {$m_i$};

\draw[<-] (7.7,0)--(7.7,-.75);
\node[l] at (7.7,-.4) {$\tilde P ( m_i || \vec r_i) = \vec y_{m_i \oplus b_{r_i} s||r_i}$};

\draw (7, -1.7) rectangle (10,-.75);
\node[l] at (7,-1.05) {$i$-th $\Enc$ query};

% QRAC message

\draw[->] (3.75,-2.5)--(7.2,-2.5)--(7.2,-1.7);
\node[fill = white, inner sep = 4pt] at (5,-2.3) {$b_{\vec r_1}, \dotsc, b_{\vec r_\ell}$};

\draw[->] (4,.9)--(7,.9);
\node at (5,1.1) {$\ket{\psi}$};

\draw[->] (4,.3)--(5.8,.3);
\node at (5,.45) {$m^*$};

% Challenge

\draw (5.8,.1) rectangle (6.2,.5);
\node at (6.0,.23) {$\Phi_j$};

\draw[->] (6.2,.3)--(7,.3);
\node at (6.5, 0.45) {$c^*$};
\node at (6.2,-0.20) {$\vec y_{m^*||j}$};

% PRF function

\node at (1.7,-2.5) {$\tilde P(\vec x || \vec r) := \begin{cases}
\vec y_{ x||  r} & \mbox{if } b_r = 0 \\
\vec y_{x \oplus  s || r} & \mbox{if } b_r = 1
\end{cases}$};

\end{tikzpicture}
\caption{Quantum random access code construction for the \PRP scheme.}\label{fig:codeP}
\end{figure}

\paragraph{Construction of a quantum random access code.} Let $\algo A$ be a \QPT adversary with winning probability $p$ and let $\ell=\poly(n)$ be an upper bound
on the number of queries made by $\algo A_2$.
When constructing a \QRAC for the \PRP scheme, we shall also assume for simplicity that both the encoder and decoder share a random permutation (as part of the shared randomness).
According to the well known \textit{coupon collector's problem}, it is sufficient for the encoder and decoder to share around $N \ln(N)$ random strings on average,
where $N$ denotes the number of distinct random strings required to make up the desired permutation.
We define a quantum random access code as follows (see also \expref{Figure}{fig:codeP}).

\begin{description}
\item[Encoding.] Let $b_1,\dots,b_{2^{n}}\in \bit$ be the string to be encoded and let the shared randomness be given by a random string $\vec s$ together with a random permutation
$\vec y = \vec y_1,\dots, \vec y_{2^{2n}}\in \bit^{2n}$ and a set of random strings $\vec r_1,\dots, \vec r_{\ell} \in \bit^n$. Using $b_1,\dots,b_{2^{n}}$, we define a new random permutation by
letting $\tilde P(x||r) := \vec y_{x \oplus b_r s||r}$ ($\tilde P$ remains a permutation\footnote{Since $\tilde P(x||r) = \tilde P(x'|| r') \iff \vec y_{x \oplus b_r s||r} = \vec y_{x' \oplus b_{r'} s||r'} \iff (r=r') \land (x=x')$}).
Run $\algo A_1$ by answering encryption and decryption queries using $\tilde{P}$ in place of $\pi$ (for decryption, use $\tilde{P}^{-1}$ and discard the last $n$ bits).
Let $m^*$ and $\ket{\psi}$ be the outputs of $\algo A_1$. Then, output $\rho=(\ket{\psi},m^*, b_{r_1}, \dots, b_{r_l})$.
\item[Decoding.] Let $j\in\{1,\dots,2^{n}\}$ be the index of the bit to be decoded; so given $\rho$ as above, we will recover $b_j$ by making use of the
shared randomness defined above.
Upon receiving a query $j\in \{1,\dots,2^{n}\}$, run $\algo A_2$ with inputs $\ket{\psi}$ and $c^* = \vec y_{m^*||j}$.
Return the bit $b'$ output by $\algo A_2$.
\end{description}

\paragraph{Average bias of the code.} We claim that the average probability of decoding correctly, taken over all choices of $b_1,\dots,b_{2^n}\in\bit$ and $j\in\{1,\dots,2^n\}$, is exactly $p$, the success probability of $\algo A$. To see this,
first note that from $\algo A$'s perspective, this is exactly \textsc{Game 1}: the function $\tilde{P}$ is a uniformly random permutation, and the queries are responded to just as in \textsc{Game 1}.
Further, note that if $b_j=0$, the challenge amounts to $\tilde P(m^*||j) = \vec y_{m^*||j}$, so the correct guess for $\algo A_2$ would be $0$, and if
$b_j=1$, then $\vec y_{x||j}$ is an encryption of a uniformly random string $\vec x =m^*\oplus \vec s$, so the correct guess for $\algo A_2$ would be 1.

Therefore, the average bias of the code is $p-1/2$. We now proceed with a similar analysis as with the \PRF scheme. Note that $\rho$ has dimension at most $2^{\poly(n)}$, since $\ket{\psi}$ must be a $\poly(n)$-qubit state ($\algo A_1$ only runs for $\poly(n)$ time), and $\ell$, the number of queries made by $\algo A_2$ must be $\poly(n)$, since $\algo A_2$ only runs for $\poly(n)$ time.
As this code encodes $2^n$ bits into a state of dimension $2^{\poly(n)}$, by \expref{Lemma}{lemma:qrac}, the bias is $O(2^{-n/2}\poly(n))=\negl(n)$, so $p\leq \frac{1}{2}+\negl(n)$. \qed
\end{proof}

\newpage

%%%%%%%%%%%%%%%%%%%%%
\section{Quantum algorithm for linear rounding functions}\label{sec:rounding}
%%%%%%%%%%%%%%%%%%%%%

\begin{figure}
\centering
\begin{tikzpicture}[l/.style={rotate=-90,anchor=west},scale=.7]

\filldraw (0,0) circle (.1);
\node[l] at (0,-.2) {$a$};

\filldraw (1,0) circle (.1);
\node[l] at (1,-.2) {$a+1$};

\node at (2,0) {$\dots$};

\filldraw (3,0) circle (.1);
\node[l] at (3,-.2) {$a+b-1$};

\draw (-.25,-.25) -- (-.25,.25) -- (3.25,.25) -- (3.25,-.25) -- (-.25,-.25);

\draw (-.25,-3) to[out=-90,in=180] (0,-3.25) -- (1.25,-3.25) to[out=0,in=90] (1.5,-3.5) to[out=90,in=180] (1.75,-3.25) -- (3,-3.25) to[out=0,in=-90] (3.25,-3);
\node at (1.5,-3.75) {$b$};

\filldraw (4,0) circle (.1);
\node[l] at (4,-.2) {$a+b$};

\node at (5.5,0) {$\dots$};

\filldraw (7,0) circle (.1);
\node[l] at (7,-.2) {$a+2b-1$};

\draw (3.75,.25) -- (3.75,-.25) -- (7.25,-.25) -- (7.25,.25) -- (3.75,.25);

\draw (3.75,-3) to[out=-90,in=180] (4,-3.25) -- (5.25,-3.25) to[out=0,in=90] (5.5,-3.5) to[out=90,in=180] (5.75,-3.25) -- (7,-3.25) to[out=0,in=-90] (7.25,-3);
\node at (5.5,-3.75) {$b$};

\node at (8.25,0) {$\dots$};

\filldraw (9.5,0) circle (.1);
\node[l] at (9.5,-.2) {$a+(c-2)b$};

\node at (11,0) {$\dots$};

\filldraw (12.5,0) circle (.1);
\node[l] at (12.5,-.2) {$a+(c-1)b-1$};

\draw (9.25,-.25) -- (9.25,.25) -- (12.75,.25) -- (12.75,-.25) -- (9.25,-.25);

\draw (9.25,-3) to[out=-90,in=180] (9.5,-3.25) -- (10.75,-3.25) to[out=0,in=90] (11,-3.5) to[out=90,in=180] (11.25,-3.25) -- (12.5,-3.25) to[out=0,in=-90] (12.75,-3);
\node at (11,-3.75) {$b$};

\filldraw (13.5,0) circle (.1);
\node[l] at (13.5,-.2) {$a+(c-1)b$};

\node at (14.5,0) {$\dots$};

\filldraw (15.5,0) circle (.1);
\node[l] at (15.5,-.2) {$a-1$};

\draw (13.25,.25) -- (13.25,-.25) -- (15.75,-.25) -- (15.75,.25) -- (13.25,.25);

\draw (13.25,-3) to[out=-90,in=180] (13.5,-3.25) -- (14.25,-3.25) to[out=0,in=90] (14.5,-3.5) to[out=90,in=180] (14.75,-3.25) -- (15.5,-3.25) to[out=0,in=-90] (15.75,-3);
\node at (14.5,-3.75) {$b-(cb-q)$};

\node at (1.5,.5) {$I_0(a,b)$};%, $\Round(\vec x)=0$};
\node at (5.5,.5) {$I_1(a,b)$};%\Round(\vec x)=1$};
\node at (11,.5) {$I_{c-2}(a,b)$};%$\Round(\vec x) = c-2$};
\node at (14.5,.5) {$I_{c-1}(a,b)$};%$\Round(\vec x) = c-1$};

\end{tikzpicture}
\caption{Dividing $\Z_q$ into $c=\lceil q/b\rceil$ blocks, starting from $a$. The first $c-1$ blocks, labelled $I_0(a,b),\dots,I_{c-2}(a,b)$, have size $b$ and the last, labelled $I_{c-1}(a,b)$, contains the remaining $b-(cb-q)\leq b$ elements of $\Z_q$.}\label{fig:blocks}
\end{figure}
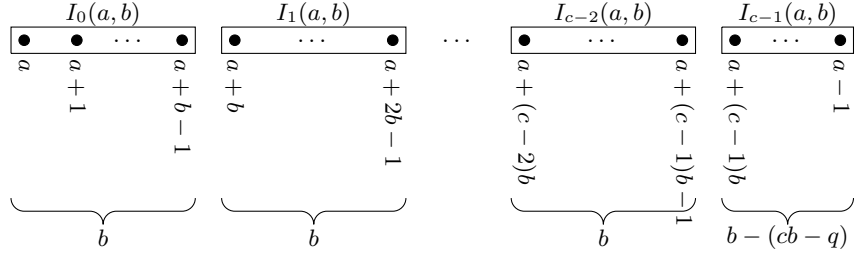

In this section, we analyze the performance of the Bernstein-Vazirani algorithm~\cite{BV97} with a modified version of the oracle. While the original oracle computes the inner product modulo $q$, our version only gives partial information about it by rounding its value to one of $\lceil q/b\rceil$ blocks of size $b$, for some $b\in\{1,\dots,q-1\}$ (if $b$ does not divide $q$, one of the blocks will have size $<b$).

\begin{definition}\label{def:LRF}
Let $n \geq 1$ be an integer and $q \geq 2$ be an integer modulus. Let $a\in\Z_q$, $b\in\Z_q\setminus\set{0}$ and $c := \ceil{q/b}$.
We partition $\Z_q$ into $c$ disjoint blocks (most of them of size $b$) starting from $a$ as follows (see \expref{Figure}{fig:blocks}):
\begin{equation*}
  I_v(a,b) :=
  \begin{cases}
    \set{a+vb,\dotsc,a+vb+b-1} & \text{if $v \in \set{0, \dotsc, c-2}$}, \\
    \set{a+vb,\dotsc,a+q-1}      & \text{if $v = c-1$}.
  \end{cases}
\end{equation*}
Based on this partition, we define a family $\Round_{\vec k,a,b}: \Z_q^n \longrightarrow \Z_c$ of keyed linear rounding functions, with key $\vec k \in \Z_q^n$, as follows:
\begin{equation*}
  \Round_{\vec{k},a,b}(\vec x) :=
  v \text{ if } \ip{\vec x, \vec k}\in I_v(a,b).
\end{equation*}
%In words, we divide the cyclic group $\Z_q$ into $c$ blocks, beginning at $a$: the first $c-1$ have size $b$, and the last one has size $b-(c b - q)\leq b$ (see \expref{Figure}{fig:blocks}). Then $\Round$ classifies $\vec x$ by which of these blocks the value $\ip{\vec x,\vec k}$ falls into.
\end{definition}

\IncMargin{1em}
\begin{algorithm}[H]
\SetKwData{Left}{left}\SetKwData{This}{this}\SetKwData{Up}{up}
\SetKwInOut{Input}{Input}\SetKwInOut{Output}{Output}\SetKwInOut{Params}{Parameters}
\Params{$n$, $q$, $b\in \{1,\dots,q-1\}$, $c=\lceil q/b\rceil$.}
\Input{Quantum oracle $U_{\Round} : \ket{\vec x} \ket{z} \mapsto \ket{\vec x} \ket{z + \Round_{\vec{k},a,b}(\vec x) \pmod c}$ where $\vec x \in \Z_q^n$, $z \in \Z_c$ and $\Round_{\vec{k},a,b}$ is the rounded inner product function for some unknown $\vec k \in \Z_q^{n}$ and $a \in \Z_q$.
}
\Output{String $\tilde{\vec{k}} \in \Z_q^n$ such that $\tilde{\vec{k}} = \vec{k}$ with high probability.}
\BlankLine
\BlankLine
\begin{enumerate}
\item Prepare the uniform superposition and append $\frac{1}{\sqrt{c}}\sum_{z=0}^{c-1}\omega_c^z\ket{z}$ where $\omega_c = e^{2\pi i/c}$:
$\displaystyle\frac{1}{\sqrt{q^{n}}}\sum_{\vec x \in \Z_q^{n}} \ket{\vec x}\otimes \frac{1}{\sqrt{c}}\sum_{z=0}^{c-1}\omega_c^z\ket{z}.$ %\phantom{--------}$
\item Query the oracle $U_{\Round}$ for $\Round_{\vec{k},a,b}$ to obtain
$\displaystyle\frac{1}{\sqrt{q^{n}}}\sum_{\vec x \in \Z_q^{n}} \omega_c^{-\Round_{\vec{k},a,b}(\vec x)} \ket{\vec x}\otimes \frac{1}{\sqrt{c}}\sum_{z=0}^{c-1}\omega_c^z\ket{z}.$ %\phantom{---}$
\item Discard the last register and apply the quantum Fourier transform $\QFT_{\Z_q}^{\otimes n}$.
\item Measure in the computational basis and output the outcome $\tilde{\vec{k}}$.
\end{enumerate}
\caption{Bernstein-Vazirani for linear rounding functions}\label{alg:key_rec_rounding}
\end{algorithm}
\DecMargin{1em}

The following theorem shows that the modulo-$q$ variant of the Bernstein-Vazirani algorithm (\expref{Algorithm}{alg:key_rec_rounding}) can recover $\vec k$ with constant probability of success by using only a single quantum query to $\Round_{\vec k, a,b}$.
\begin{theorem}%[Quantum Key-Recovery]
\label{thm:LWE-dec}
Let $U_{\Round}$ be the quantum oracle for the linear rounding function $\Round_{\vec{k},a,b}$ with modulus $q \geq 2$, block size $b\in\{1,\dots,q-1\}$, and an unknown $a\in\{0,\dots,q-1\}$, and unknown key $\vec k \in \Z_q^{n}$ such that $\vec k$ has at least one entry that is a unit modulo $q$. Let $c=\lceil q/b\rceil$ and $d=cb-q$. By making one query to the oracle $U_{\Round}$, \expref{Algorithm}{alg:key_rec_rounding} recovers the key $\vec k$ with probability at least $4/\pi^2-O(d/q)$.
%\snote{The constant has gone from $4/\pi^2$ to $(4/\pi^2)^2$. I think it would be possible to get the $4/\pi^2$ back, but can't figure out how at the moment.}
\end{theorem}

% added a blank comment
\begin{proof}

For an integer $m$, let $\omega_m=e^{2\pi i/m}$.
Several times in this proof, we will make use of the identity $\sum_{z=0}^{\ell-1}\omega_m^{rz}=\omega_m^{r(\ell-1)/2}\Bigl(\frac{\sin(\ell r \pi/m)}{\sin(r\pi/m)}\Bigr)$.

Let $c=\lceil q/b\rceil$. Throughout this proof, let $\Round(\vec x)=\Round_{\vec{k},a,b}(\vec{x})$.
By querying with $\frac{1}{\sqrt{c}}\sum_{z=0}^{c-1}\omega_c^z\ket{z}$ in the second register, we are using the standard phase kickback technique, which puts the output of the oracle directly into the phase:
\begin{eqnarray*}
\ket{\vec{x}}\frac{1}{\sqrt{c}}\sum_{z=0}^{c-1}\omega_c^z\ket{z}
& \overset{U_{\Round}}{\longmapsto}&
\ket{\vec x}\frac{1}{\sqrt{c}}\sum_{z=0}^{c-1}\omega_c^z\ket{z+\Round(\vec x)\pmod c}\\
&=& \ket{\vec x}\frac{1}{\sqrt{c}}\sum_{z=0}^{c-1}\omega_c^{z-\Round(\vec{x})}\ket{z} = \omega_{c}^{-\Round(\vec x)}\ket{\vec x}\frac{1}{\sqrt{c}}\sum_{z=0}^{c-1}\omega_c^z\ket{z}.
\end{eqnarray*}
Thus, after querying the uniform superposition over the cipherspace with $\frac{1}{\sqrt{c}}\sum_{z=0}^{c-1}\omega_c^z\ket{z}$ in the second register, we arrive at the state
\begin{equation*}
\frac{1}{\sqrt{q^{n}}} \sum_{\vec x \in \Z_q^{n}} \omega_c^{-\Round(\vec x)} \ket{\vec x}\frac{1}{\sqrt{c}}\sum_{z=0}^{c-1}\omega_c^z\ket{z}.
\end{equation*}
Note that $\omega_c = \omega_q^{q/c}$.
If we discard the last register and apply $\QFT_{\Z_q}^{\otimes n}$, we get
\begin{equation*}
  \ket{\psi} = \frac{1}{q^{n}} \sum_{\vec y \in \Z_q^{n}}
  \sum_{\vec x\in\Z_q^n}\omega_q^{-(q/c)\Round(\vec x) + \ip{\vec x, \vec y}} \ket{\vec y}.
%=\frac{1}{q^{n}} \sum_{\vec y \in \Z_q^{n}}
  %\sum_{\vec x\in\Z_q^n}\omega_q^{-b\Round(\vec x) + \ip{\vec x, \vec y} - (d/c)\Round(\vec x)} \ket{\vec y}.
\end{equation*}
We then perform a complete measurement in the computational basis. The probability of obtaining the key $\vec k$ is given by
\begin{equation}
\abs{\braket{\vec k}{\psi}}^2
= \abs*{ \frac{1}{q^{n}} \sum_{\vec x \in \Z_q^{n}} \omega_q^{-\frac{q}{c}\Round(\vec x) + \ip{\vec x, \vec k}} }^2
\!\!\!\!= \abs*{\frac{1}{q^n}\sum_{v=0}^{c-1}\omega_q^{-\frac{q}{c}v}\sum_{\vec x\in\Z_q^n:\Round(\vec x)=v}\omega_q^{\ip{\vec x,\vec k}}}^2.\label{eq:pk}
\end{equation}
We are assuming that $\vec{k}$ has at least one entry that is a unit modulo $q$. For simplicity, suppose that entry is $k_n$. Let $\vec{k}_{1:n-1}$ denote the first $n-1$ entries of $\vec k$. Then, for any $v\in\{0,\dots,c-2\}$:
\begin{eqnarray}
\sum_{\vec{x}\in\Z_q^n:\Round(\vec x)=v}\omega_q^{\ip{\vec x,\vec k}} &=& \sum_{\vec x\in\Z_q^n: \ip{\vec x,\vec k}\in I_v(a,b)}\omega_q^{\ip{\vec x,\vec k}}\nonumber\\
&=& \sum_{\vec{y}\in\Z_q^{n-1}}\omega_q^{\ip{\vec y,\vec k_{1:n-1}}}\sum_{\substack{x_n\in \Z_q: \\ x_nk_n\in I_v(a-\ip{\vec y,\vec k_{1:n-1}},b)}}\omega_q^{x_nk_n}.\label{eq:LRFv}
\end{eqnarray}
(Recall the definition of $I_v(a,b)$ from \expref{Definition}{def:LRF}).
Since $k_n$ is a unit, for each $z\in I_v(a-\ip{\vec y,\vec k_{1:n-1}})$, there is a unique $x_n\in\Z_q$ such that $x_nk_n=z$. Thus, for a fixed $\vec y\in\Z_q^{n-1}$, letting $a'=a-\ip{\vec y,\vec k_{1:n-1}}$, we have:
$$
\sum_{\substack{x_n\in \Z_q: x_nk_n\in I_v(a', b)}}\omega_q^{x_nk_n} = \sum_{z=a'+vb}^{a'+(v+1)b-1}\omega_q^z = \omega_q^{a'+vb}\sum_{z=0}^{b-1}\omega_q^z,
$$
which we can plug into \eqref{eq:LRFv} to get:
\begin{equation}
\sum_{\substack{\vec{x}\in\Z_q^n:\\\Round(\vec x)=v}}\!\!\!\omega_q^{\ip{\vec x,\vec k}}
=\!\!\!
\sum_{\vec y\in\Z_q^{n-1}}\omega_q^{\ip{\vec y,\vec k_{1:n-1}}} \omega_q^{a-\ip{\vec y,\vec k_{1:n-1}}+vb}\sum_{z=0}^{b-1}\omega_q^z = q^{n-1}\omega_q^{a+vb}\sum_{z=0}^{b-1}\omega_q^z.\label{eq:ltcase}
\end{equation}
We can perform a similar analysis for the remaining case when $v=c-1$.
Recall that $d=cb-q\geq0$ so $vb = cb-b = d+q-b = -(b-d) \pmod q$ and we get
\begin{equation}
\sum_{\vec x\in\Z_q^n:\Round(\vec x)=c-1}\omega_q^{\ip{\vec x,\vec k}}=q^{n-1}\omega_q^{a-(b-d)}\sum_{z=0}^{b-d-1}\omega_q^z.\label{eq:eqcase}
\end{equation}
This is slightly different from the $v<c-1$ case, shown in \eqref{eq:ltcase}, but very similar.
If we substitute $v=c-1$ in \eqref{eq:ltcase} and compare it to \eqref{eq:eqcase}, we get
\begin{eqnarray}
&& \abs*{q^{n-1}\omega_q^{a-(b-d)}\sum_{z=0}^{b-d-1}\omega_q^z - q^{n-1}\omega_q^{a-(b-d)}\sum_{z=0}^{b-1}\omega_q^z}\nonumber\\
&=& q^{n-1} \abs*{\sum_{z=b-d}^{b-1}\omega_q^z}
=q^{n-1}\abs*{\sum_{z=0}^{d-1}\omega_q^z}
=q^{n-1}\abs*{\frac{\sin(\pi d/q)}{\sin(\pi/q)}}\nonumber\\
&\leq & q^{n-1} \frac{\pi d/q}{2/q} = q^{n-1}\frac{\pi}{2}d.
\label{eq:approxcase}
\end{eqnarray}
Above, we have used the facts $\sin x\leq x$, and $\abs{\sin x} \geq 2x/\pi$ when $|x|\leq \pi/2$. Now, plugging \eqref{eq:ltcase} into \eqref{eq:pk} for all the $v<c-1$ terms, and using \eqref{eq:approxcase} and the triangle inequality for the $v=c-1$ term, we get:
\begin{eqnarray}
|\braket{\vec k}{\psi}| &\geq & \abs*{\frac{1}{q^n}\sum_{v=0}^{c-1}\omega_q^{-qv/c}\cdot q^{n-1}\omega_q^{a+vb}\sum_{z=0}^{b-1}\omega_q^z } - \abs*{\frac{1}{q^n}\omega_q^{-q(c-1)/c} \cdot q^{n-1}\frac{\pi}{2}d } \nonumber\\
&=& \frac{1}{q}\abs*{\sum_{v=0}^{c-1}\omega_q^{v(b-q/c)}\frac{\sin(b\pi/q)}{\sin(\pi/q)}} - \frac{\pi}{2}\frac{d}{q} \nonumber\\
&=& \frac{1}{q}\frac{\sin(b\pi/q)}{\sin(\pi/q)}\abs*{ \sum_{v=0}^{c-1}\omega_q^{v(b-q/c)} }-\frac{\pi}{2}\frac{d}{q}.\label{eq:sqrtpk}
\end{eqnarray}
Since $b-q/c=d/c$, we can bound the sum as follows:
\begin{eqnarray}
\abs*{\sum_{v=0}^{c-1}\omega_q^{v(b-q/c)}}
=\abs*{\sum_{v=0}^{c-1}\omega_q^{vd/c}}
&\geq& \abs*{\sum_{v=0}^{c-1}\cos\left(\frac{2\pi}{q}\frac{vd}{c}\right)}\nonumber\\
&\geq& \abs*{\sum_{v=0}^{c-1}\cos\left(\frac{2\pi}{q}d\right)}=\abs*{c\cos\left(\frac{2\pi d}{q}\right)}\label{eq:real1}\\
&\geq& c\sqrt{1-(2\pi d/q)^2}\label{eq:plug1}.
\end{eqnarray}
To get the inequality \eqref{eq:real1}, we used $0\leq v\leq c$ and the assumption that $d/q\leq 1/4$ (if $d/q>1/4$, the claim of the theorem is trivial), which implies that $\frac{2\pi v}{c}\frac{d}{q}\leq \frac{\pi}{2}$.
The last inequality follows from $\abs{\cos x} \geq \sqrt{1-x^2}$.

Next, we bound $\frac{\sin(b\pi/q)}{\sin(\pi/q)}$. When $b/q\leq 1/2$, $b\pi/q\leq \pi/2$, so we have $\sin(b\pi/q)\geq 2b/q$. We also have $\sin(\pi/q)\leq \pi/q$. Thus,
$$\frac{\sin(b\pi/q)}{\sin(\pi/q)}\geq \frac{2b}{\pi}.$$
On the other hand, when $b/q>1/2$, we must have $c=2$ and $b=\frac{q+d}{2}$. In that case
$$\sin(b\pi/q)=\sin\frac{\pi(q+d)}{2q}=\sin\left(\frac{\pi}{2}+\frac{\pi}{2}\frac{d}{q}\right)=\cos\frac{\pi d}{2q}\geq \sqrt{1-\left(\frac{\pi d}{2q}\right)^2}.$$
Since $\sin(\pi/q) \leq \pi/q$ and $q \geq 2b$,
\begin{equation*}
\frac{\sin(b\pi/q)}{\sin(\pi/q)}
\geq \frac{\sqrt{1-\left(\frac{\pi d}{2q}\right)^2}}{\pi/q}
\geq \frac{2b}{\pi}\sqrt{1-O(d/q)}.
\end{equation*}
Thus, in both cases, $\frac{\sin(b\pi/q)}{\sin(\pi/q)}\geq \frac{2b}{\pi}\sqrt{1-O(d/q)}$. Plugging this and \eqref{eq:plug1} into \eqref{eq:sqrtpk}, we get:
\begin{eqnarray}
|\ip{\vec{k},\psi}| &\geq & \frac{1}{q}
\cdot \frac{2b}{\pi}\sqrt{1-O(d/q)} \cdot c\sqrt{1-O(d/q)} - O(d/q)\nonumber\\
&=& \frac{2}{\pi}\frac{bc}{q} - O(d/q)
=\frac{2}{\pi}\frac{q+d}{q} - O(d/q) = \frac{2}{\pi} - O(d/q),\nonumber
\label{eq:sqrtpk2}
\end{eqnarray}
completing the proof.\qed %Finally, it is easy to see that the above contribution $O(d/q)$ includes small constant factors only, hence it vanishes for typical choices of cryptographic security parameters.
\end{proof}

%%%%%%%%%%%%%%%%%%%%%
\section{Key recovery against \LWE}\label{sec:LWE}
%%%%%%%%%%%%%%%%%%%%%

In this section, we consider various \LWE-based encryption schemes and show using \expref{Theorem}{thm:LWE-dec} that the decryption key can be efficiently recovered using a single quantum decryption query (\expref{Section}{sec:symmetric}, \expref{Section}{sec:public}, and \expref{Section}{sec:public-ring}). Then, in \expref{Section}{sec:randomness}, we show that a single quantum \emph{encryption} query can be used to recover the secret key in a symmetric-key version of $\LWE$, as long as the querying algorithm also has control over part of the randomness used in the encryption procedure.

\subsection{Key recovery via one decryption query in symmetric-key \LWE}\label{sec:symmetric}
%%%%%%%%%%%%%%%%%%%%%%%%%

Recall the following standard construction of an \INDCPA symmetric-key encryption scheme based on the \LWE assumption~\cite{Regev05}.

\begin{construction}[\LWESKES\cite{Regev05}]\label{con:LWESKES}
Let $n \geq 1$ be an integer, let $q \geq 2$ be an integer modulus and let $\chi$ be a discrete and symmetric error distribution. Then, the symmetric-key encryption scheme $\LWESKES(n,q,\chi) = (\KeyGen, \Enc, \Dec)$ is defined as follows:
\begin{enumerate}
\item \phase{\KeyGen} output $\vec k \rand \Z_q^n$;
\item \phase{$\Enc_{\vec k}$} to encrypt $b \in \bit$, sample $\vec a \rand \Z_q^n$, $e \randchi \Z_q$ and output
$(\vec a, \ip{\vec a, \vec k} + b \floor*{\frac{q}{2}} + e)$;
\item \phase{$\Dec_{\vec k}$} to decrypt $(\vec a, c)$, output $0$ if $|c - \ip{\vec a, \vec k}| \leq \floor*{\frac{q}{4}}$, else output $1$.
\end{enumerate}
\end{construction}

%We refer to an encryption scheme as \textit{correct} if $\Pr[\Dec_{\vec k}(\Enc_{\vec k}(b)) = b] = 1 - \negl(n)$. To guarantee that the above \LWESKES scheme is correct, we need to restrict the support of the discrete error distribution $\chi$ by fixing the noise magnitude $\eta$ so that $|e| \leq \floor*{\frac{q}{4}}$ for all $e \from \chi$.

As a corollary of \expref{Theorem}{thm:LWE-dec}, an adversary that is granted a single quantum decryption query can recover the key with probability at least $4/\pi^2-o(1)$:
\begin{corollary}%[Quantum Key-Recovery]
\label{thm:LWESKE-dec}
There is a quantum algorithm that makes one quantum query to $\LWESKES.\Dec_{\vec k}$ and recovers the entire key $\vec k$ with probability at least $4/\pi^2 - o(1)$.
\end{corollary}
\begin{proof}
Note that $\LWESKES.\Dec_{\vec{k}}$ coincides with a linear rounding function $\Round_{\vec{k}',a,b}$ for a key $\vec{k}' = (-\vec{k},1)\in \Z_q^{n+1}$, which has a unit in its last entry. In particular, $b=\lceil q/2\rceil$, and if $q=3\pmod 4$, $a=\lceil q/4\rceil $, and otherwise, $a=-\lfloor q/4\rfloor$. Thus, by \expref{Theorem}{thm:LWE-dec},
 \expref{Algorithm}{alg:key_rec_rounding} makes one quantum query to $\Round_{\vec k',a,b}$, which can be implemented using one quantum query to $\LWESKES.\Dec_{\vec k}$, and recovers $\vec k'$, and thus $\vec k$, with probability $4/\pi^2-O(d/q)$, where $d=\lceil q/b\rceil b-q\leq 1$.\qed
\end{proof}

Note that the key in this scheme consists of $n \log q$ uniformly random bits, and that a classical decryption query yields at most a single bit of output. It follows that any algorithm making $t$ \emph{classical} queries to the decryption oracle recovers the entire key with probability at most $2^{t - n \log q}$. A straightforward key-recovery algorithm does in fact achieve this.%using a linear number of classical queries does in fact recover the key with constant success probability.%; the details are described in \expref{Appendix}{app:classical-lwe}.

\subsection{Key recovery via one decryption query in public-key \LWE}\label{sec:public}

The key-recovery attack described in \expref{Corollary}{thm:LWESKE-dec} required nothing more than the fact that the decryption procedure of $\LWESKES$ is just a linear rounding function whose key contains the decryption key. As a result, the attack is naturally applicable to other variants of $\LWE$. In this section, we consider two public-key variants. The first is the standard construction of \INDCPA public-key encryption based on the \LWE assumption, as introduced by Regev~\cite{Regev05}. The second is the \INDCPA-secure public-key encryption scheme \FrodoPKE~\cite{FrodoKEM}, which is based on a construction of Lindner and Peikert~\cite{LP11}. In both cases, we demonstrate a dramatic speedup in key recovery using quantum decryption queries.

We emphasize once again that key recovery against these schemes was already possible classically using a linear number of decryption queries. Our results should thus not be interpreted as a weakness of these cryptosystems in their stated security setting (i.e., \INDCPA). The proper interpretation is that, if these cryptosystems are exposed to chosen-ciphertext attacks, then quantum attacks can be even more devastating than classical ones.

\paragraph{Regev's public-key scheme.}

The standard construction of an \INDCPA public-key encryption scheme based on \LWE is the following.
\begin{construction}[\LWEPKES~\cite{Regev05}]\label{cons:LWE_sym}
Let $m \geq n \geq 1$ be integers, let $q \geq 2$ be an integer modulus, and let $\chi$ be a discrete error distribution over $\Z_q$. Then, the public-key encryption scheme $\LWEPKES(n,q,\chi) = (\KeyGen, \Enc, \Dec)$ is defined as follows:
\begin{enumerate}
\item \phase{\KeyGen} output a secret key $\vec{sk} = \vec k \rand \Z_q^n$ and a public key $\vec{pk} = (\vec A,\vec A \vec k + \vec e) \in \Z_q^{m \times (n+1)}$, where $\vec A \rand \Z_q^{m \times n}$, $\vec e \randchi \Z_q^m$, and all arithmetic is done modulo $q.$
\item \phase{\Enc} to encrypt $b \in \bit$, pick a random $\vec v \in \bit^m$ with Hamming weight roughly $m/2$ and output $(\vec v\tp \vec A, \vec v\tp (\vec A\vec k + \vec e)+b\lfloor\frac{q}{2}\rfloor ) \in \Z_q^{n+1}$, where $\vec v\tp$ denotes the transpose of $\vec v$.
%\item \phase{\Enc} to encrypt $b \in \bit$ using the public key $\vec{pk}$, pick a random subset $S$ of rows of $\vec A$ of size roughly $m/2$, and output $\of*{\sum_{i \in S} \vec a_i, b \floor*{\frac{q}{2}} +  \sum_{i \in S} \of{\ip{\vec a_i, \vec k} + \vec e_i}}$;
\item \phase{\Dec} to decrypt $(\vec a, c)$, output $0$ if $|c - \ip{\vec a, \vec{sk}}| \leq \floor*{\frac{q}{4}}$, else output $1$.
\end{enumerate}
\end{construction}
Although the encryption is now done in a public-key manner, all that matters for our purposes is the decryption procedure, which is identical to the symmetric-key case, $\LWESKES$. We thus have the following corollary, whose proof is identical to that of \expref{Corollary}{thm:LWESKE-dec}:
\begin{corollary}
There is a quantum algorithm that makes one quantum query to $\LWEPKES.\Dec_{\vec{sk}}$ and recovers the entire key $\vec{sk}$ with probability at least \mbox{$4/\pi^2-o(1)$}.
\end{corollary}

\paragraph{Frodo public-key scheme.}
Next, we consider the \INDCPA-secure public-key encryption scheme \FrodoPKE, which is based on a construction by Lindner and Peikert \cite{LP11}. Compared to \LWEPKES, this scheme significantly reduces the key-size and achieves better security estimates than the initial proposal by Regev~\cite{Regev05}. For a detailed discussion of \FrodoPKE, we refer to \cite{FrodoKEM}. We present the entire scheme for completeness, but the important part for our purposes is the decryption procedure.

\begin{construction}[\FrodoPKE~\cite{FrodoKEM}]\label{cons:FrodoKEM}
Let $n,\bar m,\bar n$ be integer parameters, let $q \geq 2$ be an integer power of 2, let $B$ denote the number of bits used for encoding, and let $\chi$ be a discrete symmetric error distribution.
The public-key encryption scheme
$\FrodoPKE = (\KeyGen,\Enc,\Dec)$ is defined as follows:
\begin{enumerate}
\item \phase{\KeyGen} generate a matrix $\vec A \rand \Z_q^{n \times n}$ and matrices $\vec S,\vec E \randchi \Z_q^{n\times \bar{n}}$; compute $\vec B = \vec A \vec S + \vec E \in \Z_q^{n\times \bar{n}}$; output the key-pair $(\vec{pk},\vec{sk})$ with public key $\vec{pk} = (\vec A, \vec B)$ and secret key $\vec{sk} = \vec S$.

\item \phase{\Enc}
to encrypt $\vec m \in \bit^{B\cdot \bar m\cdot \bar n}$ (encoded as a matrix $\vec M \in \Z_{q}^{\bar{m}\times \bar{n}}$ with each entry having 0s in all but the $B$ most significant bits) with public key $\vec{pk}$,
sample error matrices $\vec{S}', \vec E' \randchi \Z_q^{\bar m \times n}$ and $\vec E'' \randchi \Z_q^{\bar m \times \bar n}$; compute $\vec C_1 = \vec S'\vec A + \vec E' \in \Z_q^{\bar m \times n}$ and $\vec C_2 =  \vec M + \vec S'\vec B + \vec E'' \in \Z_q^{\bar m \times \bar n}$;
output the ciphertext $(\vec C_1, \vec C_2)$.

\item \phase{\Dec} to decrypt $(\vec C_1, \vec C_2) \in \Z_q^{\bar{m}\times {n}}\times \Z_q^{\bar{m}\times \bar{n}}$ with secret-key $\vec{sk} = \vec S$, compute $\vec M= \vec C_2 - \vec C_1 \vec S\in \Z_q^{\bar{m}\times \bar{n}}$. For each $(i,j)\in [\bar{m}]\times [\bar{n}]$, output the first $B$ bits of $M_{i,j}$.
\end{enumerate}
\end{construction}

We now show how to recover $\bar{m}$ of the $\bar{n}$ columns of the secret key $\vec S$ using a single quantum query to $\FrodoPKE.\Dec_{\vec S}$. If $\bar{m}=\bar{n}$, as in sample parameters given in \cite{FrodoKEM}, then this algorithm recovers $\vec S$ completely.
\begin{theorem}
There exists a quantum algorithm that makes one quantum query to $\FrodoPKE.\Dec_{\vec S}$ and recovers any choice of $\bar{m}$ of the $\bar{n}$ columns of $\vec S$. For each of the chosen columns, if that column has at least one odd entry, then the algorithm succeeds in recovering the column with probability at least $4/\pi^2$.
\end{theorem}
\begin{proof}
Let $\vec s^1,\dots, \vec s^{\bar{n}}$ be the columns of $\vec S$. Let $U$ denote the map:
$$U:\ket{\vec c}\ket{z_1}\dots\ket{z_{\bar{n}}}\mapsto \ket{\vec c}\ket{z_1+\Round_{\vec s^1,0,q/2^B}(\vec c)}\dots\ket{z_{\bar{n}}+\Round_{\vec s^{\bar{n}},0,q/2^B}(\vec c)},$$
for any $\vec c\in\Z_q^n$ and $z_1,\dots, z_{\bar{n}}\in \Z_{2^B}$. We first argue that one call to $\FrodoKEM.\Dec_{\vec S}$ can be used to implement $U^{\otimes \bar{m}}$. Then we show that one call to $U$ can be used to recover any choice of the columns of $\vec S$ with probability $4/\pi^2$, as long as it has at least one entry that is odd.

Let ${\sf Trunc}:\Z_q\mapsto \Z_{2^B}$ denote the map that takes $x \in \Z_q$ to the integer represented by the $B$ most significant bits of the binary representation of $x$. We have, for any $\vec C_1\in \Z_q^{\bar{m}\times n}$, $\vec C_2=0^{\bar{m}\times \bar{n}}$, and any $\{z_{i,j}\}_{i\in[\bar{m}],j\in[\bar{n}]}\subseteq \Z_{2^B}$:
\begin{equation}
U_{\FrodoKEM.\Dec}:\ket{\vec C_1}\ket{0^{\bar m \cdot \bar n}}\bigotimes_{i\in[\bar{m}],j\in[\bar{n}]}\ket{z_{i,j}}\mapsto \ket{\vec C_1}\ket{0^{\bar m \cdot \bar n}} \bigotimes_{i\in[\bar{m}],j\in[\bar{n}]} \ket{z_{i,j} + {\sf Trunc}([\vec C_1\vec S]_{i,j})}.\label{eq:frodo1}
\end{equation}
Above, $[\vec C_1\vec S]_{i,j}$ represents the $ij$-th entry of $\vec C_1\vec S$. If $\vec c^1,\dots,\vec c^{\bar{m}}$ denote the rows of $\vec C_1$, then $[\vec C_1\vec S]_{i,j} = \ip{\vec c^i,\vec s^j}$. Thus, ${\sf Trunc}([\vec C_1\vec S]_{i,j})=\Round_{\vec s^j,0,q/2^B}(\vec c^i)$, the linear rounding function with block size $b=q/2^B$, which is an integer since $q$ is a power of 2, and $a=0$. Note that we have also assumed that the plaintext is \emph{subtracted} rather than added to the last register; this is purely for convenience of analysis, and can easily be accounted for by adjusting \expref{Algorithm}{alg:key_rec_rounding} (e.g., by using inverse-\QFT instead of \QFT.)

Discarding the second register (containing $\vec C_2 = 0$), the right-hand side of \eqref{eq:frodo1} becomes
\begin{equation}
\ket{\vec c^1}\dots\ket{\vec c^{\bar{m}}}\bigotimes_{i\in [\bar{m}],j\in[\bar{n}]}\ket{z_{i,j}+\Round_{\vec s^j,0,q/2^B}(\vec c^i)}.\label{eq:frodo2}
\end{equation}
Reordering the registers of \eqref{eq:frodo2}, we get:
$$\bigotimes_{i\in[\bar{m}]}\left(\ket{\vec c^i}\bigotimes_{j\in[\bar{n}]}\ket{z_{i,j}+\Round_{\vec s^j,0,q/2^B}(\vec c^i)}\right)
= U^{\otimes \bar m}\left(\bigotimes_{i\in [\bar{m}]}\ket{\vec c^i}\bigotimes_{j\in[\bar{n}]}\ket{z_{i,j}} \right).$$
Thus, we can implement $U^{\otimes \bar m}$ using a single call to $\FrodoKEM.\Dec_{\vec S}$.

Next we show that for any particular $j\in [\bar{n}]$, a single call to $U$ can be used to recover $\vec{s}^j$, the $j$-th column of $\vec S$, with probability at least $4/\pi^2$, as long as at least one entry of $\vec s^j$ is odd. To do this, we show how one use of $U$ can be used to implement one phase query to $\Round_{s^j,0,q/2^B}$. Then the result follows from the proof of \expref{Theorem}{thm:LWE-dec}.

Let $\ket{\varphi} = 2^{-B/2} \sum_{z=0}^{2^B-1} \ket{z}$, and define

$$
\ket{\phi_j} = \ket{\varphi}^{\otimes (j-1)} \otimes \frac{1}{\sqrt{2^B}}\sum_{z=0}^{2^B-1}\omega_{2^B}^z\ket{z}\otimes \ket{\varphi}^{\otimes (\bar{n}-j)}.$$
Then for any $\vec{c}\in \Z_q^n$, we have:
$$
\frac{1}{\sqrt{2^B}}\sum_{z=0}^{2^B-1}\ket{z+\Round_{\vec s^i,0,q/2^B}(\vec c)} = \frac{1}{\sqrt{2^B}} \sum_{z=0}^{2^B-1} \ket{z} = \ket{\varphi},
$$
since addition here is modulo $2^B$, and
$$
\frac{1}{\sqrt{2^B}}\sum_{z=0}^{2^B-1}\omega_{2^B}^z\ket{z+\Round_{\vec s^j,0,q/2^B}(\vec c)}=\frac{1}{\sqrt{2^B}}\sum_{z=0}^{2^B-1}\omega_{2^B}^{z-\Round_{\vec s^j,0,q/2^B}(\vec c)}\ket{z}.$$
Thus:
\begin{eqnarray*}
U(\ket{\vec c}\ket{\phi_j}) &=& \ket{\vec c}\ket{\varphi}^{\otimes (j-1)} \otimes \frac{1}{\sqrt{2^B}}\sum_{z=0}^{2^B-1}\omega_{2^B}^{z-\Round_{\vec s^j,0,q/2^B}(\vec c)}\ket{z}\otimes \ket{\varphi}^{\otimes (\bar{n}-j)}\\
&=& \omega_{2^B}^{-\Round_{\vec s^j,0,q/2^B}(\vec c)}\ket{\vec c}\ket{\phi_j}.
\end{eqnarray*}
Thus, by the proof of \expref{Theorem}{thm:LWE-dec}, if we apply $U$ to $q^{-n/2} \sum_{\vec c\in\Z_q^n}\ket{\vec c}\ket{\phi_j}$, Fourier transform the first register, and then measure, assuming $\vec{s}^j$ has at least one entry that is a unit\footnote{since $q$ is a power of 2, this is just an odd number} we will measure $\vec{s}^j$ with probability at least $\pi^2/4-O(d/q)$, where $d=q/2^B\lceil q/(q/2^B)\rceil - q = 0$.%, since $q$ is a power of 2.

Thus, if we want to recover columns $j_1,\dots j_{\bar{m}}$ of $\vec S$, we apply our procedure for $U^{\otimes \bar{m}}$, which costs one query to $\FrodoKEM.\Dec_{\vec S}$, to the state
$$\sum_{\vec c\in\Z_q^n}\frac{1}{\sqrt{q^n}}\ket{\vec c}\ket{\phi_{j_1}}\otimes \dots\otimes
\sum_{\vec c\in\Z_q^n}\frac{1}{\sqrt{q^n}}\ket{\vec c}\ket{\phi_{j_{\bar{m}}}},$$
Fourier transform each of the $\vec c$ registers, and then measure.\qed
\end{proof}

%$$\ket{\phi_j} = \left(\frac{1}{\sqrt{2^B}}\sum_{z=0}^{2^B-1}\ket{z}\right)^{\otimes (j-1)} \otimes \frac{1}{\sqrt{2^B}}\sum_{z=0}^{2^B-1}\omega_{2^B}^z\ket{z}\otimes \left(\frac{1}{\sqrt{2^B}}\sum_{z=0}^{2^B-1}\ket{z}\right)^{\otimes (\bar{n}-j)}.$$
%Then for any $\vec{c}\in \Z_q^n$, we have:
%$$\frac{1}{\sqrt{2^B}}\sum_{z=0}^{2^B-1}\ket{z+\Round_{\vec s^i,0,q/2^B}(\vec c)} = \frac{1}{\sqrt{2^B}}\sum_{z=0}^{2^B-1}\ket{z},$$
%since addition here is modulo $2^B$,
%and
%$$\frac{1}{\sqrt{2^B}}\sum_{z=0}^{2^B-1}\omega_{2^B}^z\ket{z+\Round_{\vec s^j,0,q/2^B}(\vec c)}=\frac{1}{\sqrt{2^B}}\sum_{z=0}^{2^B-1}\omega_{2^B}^{z-\Round_{\vec s^j,0,q/2^B}(\vec c)}\ket{z}.$$
%Thus:
%\begin{eqnarray*}
%U(\ket{\vec c}\ket{\phi_j}) &=& \ket{\vec c}\left(\frac{1}{\sqrt{2^B}}\sum_{z=0}^{2^B-1}\ket{z}\right)^{\otimes (j-1)} \otimes \frac{1}{\sqrt{2^B}}\sum_{z=0}^{2^B-1}\omega_{2^B}^{z-\Round_{\vec s^j,0,q/2^B}(\vec c)}\ket{z}\otimes \left(\frac{1}{\sqrt{2^B}}\sum_{z=0}^{2^B-1}\ket{z}\right)^{\otimes (\bar{n}-j)}\\
%&=& \omega_{2^B}^{-\Round_{\vec s^j,0,q/2^B}(\vec c)}\ket{\vec c}\ket{\phi_j}.
%\end{eqnarray*}

\subsection{Key recovery via one decryption query in public-key Ring-\LWE}\label{sec:public-ring}

Next, we analyze key-recovery with a single quantum decryption query against Ring-\LWE encryption. Unlike the plain \LWE-based encryption schemes we considered in the previous sections, Ring-\LWE encryption uses noisy samples over a polynomial ring. In the following, we consider the basic, bit-by-bit Ring-\LWE public-key encryption scheme introduced in \cite{LPR-toolkit-2013,LPR13}. It is based on the rings $\mathcal{R}= \Z[x]/\ip{x^n+1}$ and $\mathcal{R}_q := \mathcal{R}/q \mathcal{R} = \Z_q[x]/\ip{x^n+1}$ for some power-of-two integer $n$ and $\poly(n)$-bounded prime modulus $q$. The details of the error distribution $\chi$ below will not be relevant to our results.

%In order to generate samples, we assume a symmetric error distribution $\chi$ that samples "small" error polynomials from $\mathcal R$. For example, a typical choice (under a particular representation for $\mathcal R_q$) is to use an $n$-dimensional Gaussian, more specifically the product of $n$ centered one-dimensional Gaussian distributions.

%Let us first recall the basic public-key encryption scheme based on Ring-\LWE. We restrict our analysis to single-bit encryption only.

\begin{construction}[Ring-\LWE-\PKE \cite{LPR-toolkit-2013,LPR13}]\label{cons:Ring-LWE}
Let $n \geq 1$ be an integer, let $q \geq 2$ be an integer modulus, and let $\chi$ be an error distribution over $\mathcal{R}$.
The public-key encryption scheme
\RingLWE-\PKE $= (\KeyGen,\Enc,\Dec)$ is defined as follows:
\begin{enumerate}
\item \KeyGen: sample $a \rand \mathcal{R}_q$ and $e,s \randchi \mathcal R$; output $\sk = s$ and $\pk = (a,c=a\cdot s+e \pmod q)\in \mathcal{R}_q^2$.
\item \Enc: to encrypt $b \in \bit$, sample $r,e_1,e_2 \randchi \mathcal{R}$ and
output a ciphertext pair $(u,v) \in \mathcal{R}_q^2$, where
$u= a\cdot r + e_1 \pmod q$ and $v = c \cdot r + e_2 + b \floor{q/2} \pmod q$.
\item \Dec: to decrypt $(u,v)$, compute $v-u\cdot s = (r \cdot e - s \cdot e_1 + e_2) + b \floor{q/2} \pmod q \, \in \mathcal{R}_q$; output $0$ if the constant term of the polynomial is closer to $0$ than $\floor{q/2}$, else output $1$.
\end{enumerate}
\end{construction}

We note that our choice of placing single-bit encryption in the constant term of the polynomial is somewhat arbitrary. Indeed, it is straightforward to extend our results to encryption with respect to other monomials.
%To speed up multiplication of large polynomials, various embeddings are adopted in practice that represent elements in $\mathcal{R}_q$ as vectors in $\Z_q^n$ \cite{LPR-toolkit-2013}. For example, one can adopt the number-theoretic transform (NTT) to reduce polynomial multiplication to the much faster component-wise multiplication, such as in the post-quantum proposal \textit{New Hope} \cite{DPS16}. 
%We emphasize that our results are independent of the actual embedding used in practice, as long as the representation is an isomorphism between $\mathcal R_q$ and $\Z_q^n$. The same is true for the choice of ring $\mathcal R$ under a cyclotomic polynomial, and ultimately results in only slightly different classical post-processing.
We show the following corollary to \expref{Theorem}{thm:LWE-dec}.%, an adversary that is allowed to make a single quantum decryption query can recover the key with probability at least $4/\pi^2-o(1)$:
\begin{corollary}%[Quantum Key-Recovery]
\label{thm:LWESKE-dec}
There is a quantum algorithm that makes one quantum query to \RingLWE-$\PKE.\Dec_{s}$ and recovers the entire key $s$ with probability at least $4/\pi^2 - o(1)$.
\end{corollary}
\begin{proof}
We first analyze the decryption function. Let $(p)_0$ denote the constant term of a polynomial $p \in \mathcal R_q$. Then, for any two polynomials $u = \sum_{j=0}^{n-1} u_j x^j$ and $s = \sum_{j=0}^{n-1} s_j x^j \in \mathcal R_q$, we can identify the constant term of $u \cdot s$ as
\begin{equation}\label{eq:ring_lwe_ip}
(u\cdot s)_0 = u_0 s_0 + \sum_{j=1}^{n-1} u_j s_{n-j} x^j x^{n-j} \equiv u_0 s_0 - u_1 s_{n-1} - u_2 s_{n-2} - \hdots - u_{n-1} s_1 \pmod q,
\end{equation}
since $x^n \equiv -1$ in $\mathcal R_q$. We show that the outcome of \RingLWE-$\PKE.\Dec_{s}(u,v)$ coincides with a binary linear rounding function over $\Z_q^n$. Let $\vec u,\vec s \in \Z_q^n$ denote the coefficient vectors of $u,s \in \mathcal R_q$ respectively, and define a constant polynomial $v \equiv v_0 \in \mathcal R_q$ and vector $\vec u' := (\vec u,v_0) \in \Z_q^{n+1}$, for some $v_0 \in \mathcal \Z_q$. Consequently, \RingLWE-$\PKE.\Dec_{s}(u,v_0)$ rounds the inner product $\ip{\vec u',\vec s'}$, where $\vec s' = (- s_0 \pmod q, s_{n-1},\hdots,s_1, 1)$. Thus, we can run the Bernstein-Vazirani algorithm for binary linear rounding functions on a uniform superposition over $\Z_q^n$ and recover $\vec s$ from $\vec s'$ with simple classical post-processing. Note also that any choice of isomorphism between $\mathcal R_q$ and $\Z_q^n$ necessarily preserves the inner product in Eq.\eqref{eq:ring_lwe_ip}, and thus any measurement outcome can be mapped back to the standard basis prior to post-processing -- independently of the actual ring representation used in practice. By \expref{Theorem}{thm:LWE-dec},
 \expref{Algorithm}{alg:key_rec_rounding} makes one quantum query to $\Round_{\vec s',q}$, which can be implemented using one quantum query to \RingLWE-$\PKE.\Dec_{s}$, and recovers $\vec s'$, and thus $\vec s$, with probability $4/\pi^2-o(1)$.\qed
\end{proof}

\subsection{Key recovery via a randomness-access query}\label{sec:randomness}
%%%%%%%%%%%%%%%

While a linear number of classical decryption queries can be used to break $\LWE$-based schemes, we have shown that only a single \emph{quantum} decryption query is required. A natural question to ask is whether a similar statement can be made for \emph{encryption} queries. Classically, it is known that the symmetric key version of \LWE described in \expref{Construction}{con:LWESKES}, $\LWESKES$, can be broken using a linear number of classical encryption queries when the adversary is also allowed to choose the randomness used by the query: the adversary simply uses $e=0$ each time, with $\vec{a}$ taking $n$ linearly independent values. In case the adversary is allowed to make quantum encryption queries with randomness access, a \emph{single} quantum query suffice to recover \emph{the entire key} with non-negligible probability, even when the adversary only has control over \emph{a part of} the randomness used by the encryption: the randomness used to prepare vectors $\vec a$, but not the randomness used to select the error $e$. Specifically, the adversary is given quantum oracle access to the \textit{randomness-access encryption oracle} $U_{\Enc_{\vec k}}^{\RA}$ such that, on input $(b; \vec a)$, the adversary receives
$$
\Enc_{\vec k}^{\RA}(b; \vec a)=(\vec a,\langle \vec a, \vec k \rangle + b \left \lfloor{{q}/{2}}\right \rfloor+ e),
$$
where $e \from \chi$. We extend this to a quantum randomness-access oracle by answering each element of the superposition using i.i.d.\ errors $e_a \from \chi$:
$$
U_{\Enc_{\vec k}}^{\RA} :\ket{m} \ket{a} \ket{c} \mapsto \ket{m} \ket{a} \ket{c \oplus \Enc_{\vec k}^{\RA}(m ; a)}.
$$
This model is identical to the noisy learning model considered by Grilo et al.~\cite{GK17} and thus matches
the original proposal by Bshouty and Jackson~\cite{BJ95}.

First, it is not hard to see that algorithms making classical queries to the above oracle can extract at most $\log q$ bits of key from each query (specifically, from the last component of the ciphertext), and thus still require a linear number of queries to recover the complete key with non-negligible probability.

On the other hand, by a slight generalization of the proof of Theorem IV$.1$ from Ref.~\cite{GK17}, we can recover the entire key with inverse polynomial success probability using a single query to $U_{\Enc_{\vec k}}^{\RA}$ as long as the noise magniture $\eta$ is polynomial in $n$, since $\varphi(q)=\Omega(q/\log\log q)$, for Euler's totient function $\varphi$.

\begin{theorem}%[Quantum Key-Recovery]
Consider $\LWESKES(n,q,\chi)$ with an arbitrary integer modulus $\,2 \leq q \leq \exp(n)$ and a symmetric error distribution $\chi$ of noise magnitude $\eta$. Then, %\expref{Algorithm}{alg:key_rec_randomness}
there exists a quantum algorithm that
makes one query to a randomness-accessible quantum encryption oracle for $\LWESKES(n,q,\chi)$ and recovers the entire key with probability at least $\varphi(q)/(24\eta q)-o(1)$.
\end{theorem}

Finally, in a different model in which a single error $e\leftarrow \chi$ is used for every branch of the superposition of a single query (independent of $a$) we can recover $\vec k$ using a single query to the randomness access encryption oracle: simply query $\ket{0}\frac{1}{\sqrt{q^n}}\sum_{\vec a\in \Z_q^n}\ket{\vec a} \frac{1}{\sqrt{q}}\sum_{z=0}^{q-1}\omega_q^z\ket{z}$ to get $\ket{0}\frac{1}{\sqrt{q^n}}\sum_{\vec a\in \Z_q^n}\omega_q^{-\ip{\vec a,\vec k}}\ket{\vec a} \frac{1}{\sqrt{q}}\sum_{z=0}^{q-1}\omega_q^z\ket{z+e}$, apply the quantum Fourier transform to the second register, and then measure the second register to get $\vec k$ with probability 1.

%%%%%%%%%%%%%%%%%%%%%%%%%%%%%%%%%%%%
%\section{Discussion}
%%%%%%%%%%%%%%%%%%%%%%%%%%%%%%%%%%%%

%\ga{We have lots of space, so feel free to add a discussion, or to promote some things from the appendix to the main body (CCA1 equivalence, the QRAC bound?)}

%%%%%%%%%%%%%%%%%%%%%%%%%%%%%%%%%%%%
%\section{Acknowledgements}
%%%%%%%%%%%%%%%%%%%%%%%%%%%%%%%%%%%%%
%
%We thank Ronald de Wolf for helpful discussions
%and Jop Bri\"et for the proof of \expref{Lemma}{lem:QRAC}.
%SJ is supported by an NWO WISE Grant
%and NWO Veni Innovational Research Grant under project number 639.021.752. AP is partially supported by AFOSR YIP award number FA9550-16-1-0495 and the Institute for Quantum Information and Matter, an NSF Physics Frontiers Center (NSF Grant PHY-1733907).

%---------------------------------------------------------------------%
%\bibliographystyle{abbrv}
% \bibliographystyle{plainnat}

% \bibliographystyle{alphaurl}
% \bibliography{references,alagic-full-bib,bibdb}

\printbibliography
%\bibliographystyle{splncs04}
%\bibliography{eurocrypt-refs}

%---------------------------------------------------------------------%

\appendix

%%%%%%%%%%%%%%%%%%%%%%
\section{Appendix}\label{app:ota}
%%%%%%%%%%%%%%%%%%%%%%

\subsubsection{Bound for quantum random access codes.}\label{sec:qrac}
%%%%%%%%%%%%%%%%%%%%%%

Recall that a \emph{quantum random access code} (\QRAC) is the following scenario involving two parties, Alice and Bob~\cite{Nayak99}:
\begin{itemize}
  \item Alice receives an $N$-bit string $x$ and encodes it as a quantum state $\rho_x$.
  \item Bob receives $\rho_x$ from Alice and is asked to recover the $i$-th bit of $x$, for some $i \in \{1, \dotsc, N\}$, by measuring the state.
  \item They win if Bob's output agrees with $x_i$ and lose otherwise.
\end{itemize}
A variation of this scenario allows Alice and Bob to use \emph{shared randomness} in their encoding and decoding operations \cite{ALMO08} (note that shared randomness \textit{per se} does not allow them to communicate).
% If we denote the shared random variable by $\lambda$, Alice's  can now produce $\rho^\lambda_x$

We are interested in bounding the average bias $\epsilon = p_\text{win} - 1/2$ of a quantum random access code with shared randomness, where $p_\text{win}$ is the winning probability averaged over $x \inrand \bit^N$ and $i \inrand \{1, \dotsc, N\}$.

\begin{lemma}\label{lem:QRAC}
The average bias of a quantum random access code with shared randomness that encodes $N$ bits into a $d$-dimensional quantum state is $O(\sqrt{N^{-1} \log d})$.
In particular, if $N = 2^n$ and $d = 2^{\poly(n)}$ the bias is $O(2^{-n/2}\poly(n))$.
\end{lemma}

\begin{proof}
A quantum random access code with shared randomness that encodes $N$ bits into a $d$-dimensional quantum state is specified by the following:
\begin{itemize}
  \item a shared random variable $\lambda$,
  \item for each $x \in \bit^N$, a $d$-dimensional quantum state $\rho^\lambda_x$ encoding $x$,
  \item for each $i \in \{0, \dotsc, N\}$, an observable $M^\lambda_i$ for recovering the $i$-th bit.
\end{itemize}
Formally, $\rho^\lambda_x$ and $M^\lambda_i$ are $d \times d$ Hermitian matrices such that $\rho^\lambda_x \geq 0$, $\tr \rho^\lambda_x = 1$, and $\norm{M^\lambda_i} \leq 1$ where $\norm{M^\lambda_i}$ denotes the operator norm of $M^\lambda_i$. Note that both $\rho^\lambda_x$ and $M^\lambda_i$ depend on the shared random variable $\lambda$, meaning that Alice and Bob can coordinate their strategies.

The bias of correctly guessing $x_i$, for a given $x$ and $i$, is 
$
(-1)^{x_i} \tr (\rho^\lambda_x M^\lambda_i) / 2.
$
If the average bias of the code is $\epsilon$ then
$
\E_\lambda
  \E_{x,i}
  (-1)^{x_i} \tr (\rho^\lambda_x M^\lambda_i)
  \geq 2 \epsilon.
$
We can rearrange this expression and upper bound each term using its operator norm, and then apply the noncommutative Khintchine inequality \cite{Tomczak1974}:
\begin{align*}
  \E_\lambda
  \E_x
  \frac{1}{N}
  \tr \Bigl( \rho^\lambda_x \sum_{i=1}^N (-1)^{x_i} M^\lambda_i \Bigr)
  & \leq
  \E_\lambda
  \E_x
  \frac{1}{N}
  \norm{ \sum_{i=1}^N (-1)^{x_i} M^\lambda_i } \\
  & \leq
  \E_\lambda
  \frac{1}{N}
  c \sqrt{N \log d}
    =
  c \sqrt{\frac{\log d}{N}},
\end{align*}
for some constant $c$. In other words,
\begin{equation*}
  \epsilon \leq \frac{c}{2} \sqrt{\frac{\log d}{N}}.
\end{equation*}
In the particular case we are interested in, $d = 2^{\poly(n)}$ and $N = 2^n$ so
\begin{equation*}
  \epsilon \leq \frac{c}{2} \sqrt{\frac{\poly(n)}{2^n}},
\end{equation*}
completing the proof.\qed
\end{proof}

\subsubsection{Equivalence of \QCCA models.}\label{sec:cca1-equiv}
%%%%%%%%%%%%%%%%%%%%%%

%Let $\Pi = (\KeyGen, \Enc, \Dec)$ be an encryption scheme, $\algo A$ a \QPT, and $n$ the security parameter.
%Recall the definition of $\IndGameR(\Pi, \algo A, n)$:
%\begin{enumerate}
%\item \phase{Setup} A key $k \from \KeyGen(1^n)$ and a bit $b \inrand \bit$ are generated;
%\item \phase{Pre-challenge} $\algo A$ receives oracles $\Enc_k$ and $\Dec_k$, and outputs $m$;
%\item \phase{Challenge} if $b=0$, $\algo A$ receives $\Enc_k(m)$; if $b=1$, $\algo A$ receives $\Enc_k(r)$ for uniformly random $r$; $\algo A$ also receives oracle for $\Enc_k$ only; $\algo A$ outputs a bit $b'$;
%\item \phase{Resolution} $\algo A$ wins if $b = b'$.
%\end{enumerate}

Recall that the \INDQCCA notion is based on the security game \IndGame defined in \expref{Definition}{def:ind_qcca}. In the alternative security game \IndGameR (see \expref{Definition}{def:ind_qcca_r}), the adversary provides only one plaintext $m$ and must decide if the challenge is an encryption of $m$ or an encryption of a random string. In this section, we prove the following:
\begin{proposition}\label{prop:cca1-equiv}
An encryption scheme $\Pi$ is \INDQCCA if and only if for every \QPT $\algo A$,
$$
\Pr[\algo A \text{ wins }\IndGameR(\Pi, \algo A, n)] \leq 1/2 + \negl(n)\,.
$$
\end{proposition}
\begin{proof}
%We can then define security by asking that, for all $\QPT$s $\algo A$, the success probability of $\algo A$ at $\IndGameR(\Pi, \algo A, n)$ is at most $1/2 + \negl(n)$. We will show that this definition is equivalent to \INDQCCA.

Fix a scheme $\Pi$. For one direction, suppose $\Pi$ is \INDQCCA and let $\algo A$ be an adversary against \IndGameR. Define an adversary $\algo A_0$ against \IndGame as follows: (i.) run $\algo A$ until it outputs a challenge plaintext $m$, (ii.) sample random $r$ and output $(m, r)$, (iii.) run the rest of $\algo A$ and output what it outputs. The output distribution of $\IndGameR(\Pi, \algo A, n)$ is then identical to $\IndGame(\Pi, \algo A_0, n)$, which in turn must be negligibly close to uniform by \INDQCCA security of $\Pi$.

For the other direction, suppose no adversary can win \IndGameR with probability better than $1/2$, and let $\algo B$ be an adversary against \IndGame. Now, define two adversaries $\algo B_0$ and $\algo B_1$ against \IndGameR as follows. The adversary $\algo B_c$ does: (i.) run $\algo B$ until it outputs a challenge $(m_0, m_1)$, (ii.) output $m_c$, (iii.) run the rest of $\algo B$ and output what it outputs. Note that the pre-challenge algorithm is identical for $\algo B$, $\algo B_0$, and $\algo B_1$; define random variables $M_0$, $M_1$ and $R$ given by the two challenges and a uniformly random plaintext, respectively. The post-challenge algorithm is also identical for all three adversaries; call it $\algo C$. The advantage of $\algo B$ over random guessing is then bounded by
\begin{align*}
&\|\algo C(\Enc_k(M_0)) - \algo C(\Enc_k(M_1))\|_1\\
&~~~= \|\algo C(\Enc_k(M_0)) - \algo C(\Enc_k(M_1)) - \algo C(\Enc_k(R)) + \algo C(\Enc_k(R)) \|_1\\
&~~~\leq \|\algo C(\Enc_k(M_0)) - \algo C(\Enc_k(R))\|_1 + \|\algo C(\Enc_k(M_1)) - \algo C(\Enc_k(R)) \|_1\\
&~~~\leq \negl(n)\,,
\end{align*}
where the last inequality follows from our initial assumption, applied to both $\algo B_0$ and $\algo B_1$. It follows that $\Pi$ is \INDQCCA.\qed
\end{proof}

%---------------------------------------------------------------------%
%---------------------------------------------------------------------%
\end{document}
%---------------------------------------------------------------------%
%---------------------------------------------------------------------%

\IncMargin{1em}
\begin{algorithm}[H]
\SetKwData{Left}{left}\SetKwData{This}{this}\SetKwData{Up}{up}
\SetKwInOut{Input}{input}\SetKwInOut{Output}{output}
\Input{Quantum oracle $U_{\Round} : \ket{\vec x} \ket{b} \mapsto \ket{\vec x} \ket{b \oplus \Round_{\vec{k},q}(\vec x)}$ for a linear rounding function $\Round_{\vec{k},q}$ with modulus $q$ and an unknown key $\vec k \in \Z_q^{n}$ with at least one unit modulo $q$.
}
\Output{String $\tilde{\vec{k}} \in \Z_q^n$ such that $\tilde{\vec{k}} = \vec{k}$ with probability at least $4/\pi^2-o(1)$.}
\BlankLine
\BlankLine
\begin{enumerate}
\item Prepare a uniform superposition and append $(\ket{0}-\ket{1})/{\sqrt{2}}$, resulting in
$$\frac{1}{\sqrt{q^{n}}}\sum_{\vec x \in \Z_q^{n}} \ket{\vec x}\otimes \textstyle \frac{\ket{0}-\ket{1}}{\sqrt{2}}. \phantom{--------}$$
\item Query the oracle $U_{\Round}$ for $\Round_{\vec{k},q}$ to obtain
$$\frac{1}{\sqrt{q^{n}}}\sum_{\vec x \in \Z_q^{n}} (-1)^{\Round_{\vec{k},q}(\vec x)} \ket{\vec x}\otimes \textstyle \frac{\ket{0}-\ket{1}}{\sqrt{2}}. \phantom{---}$$
\item Discard the last register and apply the quantum Fourier transform $\QFT_{\Z_q}^{\otimes n}$.
\item Measure in the computational basis and output the outcome $\ket{\tilde k_1}\hdots \ket{\tilde k_{n}}$.
\end{enumerate}
\caption{Bernstein-Vazirani for binary linear rounding functions}\label{alg:key_rec_rounding}
\end{algorithm}
\DecMargin{1em}